\documentclass[aps,prd,onecolumn,nofootinbib, 11pt,eqsecnum]{revtex4-2}

\usepackage{amsmath, amssymb}
\usepackage{graphicx}

\usepackage{lineno}
\usepackage{amsthm}
    \theoremstyle{definition}
    \newtheorem{dfn}{Definition}[section]
    \newtheorem{prop}[dfn]{Proposition}

    \newtheorem{rem*}[dfn]{Remark}

\usepackage[T1]{fontenc}
\usepackage{stix}
\def\dd{\mathrm{d}}
\def\ee{\mathrm{e}}

\usepackage{mathtools}

\usepackage{tikz}

\usepackage{ulem}
\newcommand*\circled[1]{\textcircled{\footnotesize#1}}

\usepackage{hyperref}

\newcommand{\Res}[2]{\underset{#1\, =\, #2}{\text{Res}}}

\renewcommand{\emph}[1]{{\it #1}}

\begin{document}

\title{Path to an exact WKB analysis of black hole quasinormal modes}



\author{Taiga Miyachi}
\affiliation{Department of Physics, Kobe University, Kobe 657-8501, Japan}
\thanks{email: \href{tmiyachi@omu.ac.jp}{tmiyachi@omu.ac.jp}}

\author{Ryo Namba}
\affiliation{RIKEN Interdisciplinary Theoretical and Mathematical Sciences (iTHEMS), Wako, Saitama 351-0198, Japan}
\thanks{email: \href{ryo.namba@riken.jp}{ryo.namba@riken.jp}}

\author{Hidetoshi Omiya}
\affiliation{Department of Physics$,$ Kyoto University$,$ Kyoto 606-8502$,$ Japan}
\thanks{email: \href{omiya@tap.scphys.kyoto-u.ac.jp}{omiya@tap.scphys.kyoto-u.ac.jp}}

\author{Naritaka~Oshita}
\affiliation{Center for Gravitational Physics and Quantum Information,
Yukawa Institute for Theoretical Physics, Kyoto University, 606-8502, Kyoto, Japan}
\affiliation{The Hakubi Center for Advanced Research, Kyoto University,
Yoshida Ushinomiyacho, Sakyo-ku, Kyoto 606-8501, Japan}
\affiliation{RIKEN Center for Interdisciplinary Theoretical and Mathematical Sciences (iTHEMS), Wako, Saitama 351-0198, Japan}
\thanks{email: \href{naritaka.oshita@yukawa.kyoto-u.ac.jp}{naritaka.oshita@yukawa.kyoto-u.ac.jp}}


\begin{abstract}
We investigate black hole quasinormal modes using the exact WKB method.
We perform an analytic continuation from the horizon to infinity along the positive real axis of the radial coordinate and impose appropriate boundary conditions at these asymptotic positions.
We clarify the role of previously overlooked logarithmic spirals of Stokes curves and branch cuts emerging from the horizon.
We carefully reformulate the derivation of the quasinormal mode conditions using the exact WKB analysis, incorporating the contributions from these features into the calculation.
We successfully derive correct results for both solvable model examples and the Schwarzschild spacetime. Our formulation enjoys straightforward extensions to other background geometries as well as a wide range of other physical systems.
\end{abstract}

\nolinenumbers 

\unitlength = 1mm
\begin{flushright}
KOBE-COSMO-25-03, YITP-25-20
\end{flushright}

\maketitle

\newpage
\section{Introduction}
\label{sec:Intro}

Black holes are one of the best suites for testing gravity in a strong field regime, as astrophysical black holes can be characterized solely by two parameters, i.e., their mass and angular momentum, in general relativity. Such uniqueness of black holes can be seen in the signal sourced by a ringing black hole, the so-called ringdown phase. A ringdown waveform can be decomposed into quasinormal modes (QNMs) and power-law tails. The latter may be less dominant in most cases. The QNM frequencies of the Schwarzschild or Kerr black hole take complex values, reflecting the dissipative nature of black holes. The detailed nature of QNM spectra has been actively investigated based on the black hole perturbation theory.
We should also emphasize that the QNM spectrum has been discussed in various contexts in theoretical physics, such as astrophysics, quantum gravity, and its holographic interpretation (see Refs. \cite{Berti:2009kk,Konoplya:2011qq} and references therein).

Mathematically, QNMs are the complex eigenvalues of a second-order ordinary differential equation with two-point boundary conditions.
Some numerical techniques, e.g.~Leaver's method \cite{Leaver:1985ax,Leaver:1986vnb,Leaver:1986gd}, have been developed to solve the problem. However, it is still a hard problem to understand the global structure of QNM spectra over the complex frequency plane analytically. The main difficulty is to find the QNM eigenfunctions in a systematic manner.
One approach involves constructing approximate solutions at both boundaries and matching them through analytic continuation.
In this process, QNMs are encoded in the coefficients of the solutions in the matching region, analogous to the Bohr-Sommerfeld quantization rule in quantum mechanics. This enables us to find the global pattern of the QNM distribution, especially at the eikonal limit, $\ell \to \infty$, and for the highly damping modes, $\text{Im} \, (\omega) \to - \infty$.
The (Jeffereys-)Wentzel–Kramers–Brillouin (WKB) solutions are widely used to construct such approximations and perform analytic continuation \cite{Zouros:1979iw,Detweiler:1980uk,Schutz:1985km,Ferrari:1984zz,Mashhoon:1985cya,Decanini:2009mu,Andersson:2003fh,Iyer:1986np,Konoplya:2003ii,Matyjasek:2017psv,Konoplya:2019hlu}\footnote{Another analytic methods using gauge theoretical knowledge have been developed in \cite{Aminov:2020yma,Bonelli:2021uvf,Novaes:2014lha}.}.

Over the past few decades, the analysis based on a full-order resummation of the WKB infinite series, known as the {\it exact WKB analysis}, has been extensively developed \cite{Voros1983TheRO, 10.1007/978-4-431-68170-0_1, AIF_1993__43_1_163_0,2024arXiv241017224N,NIKOLAEV_2022,Nikolaev_2023}.
See also Refs. \cite{kawai2005algebraic,2014arXiv1401.7094I} for reviews.
The standard WKB solutions are often comprised of divergent series in physically interesting systems. As a result, while they can still serve as asymptotic expansions of the true solutions, they are unable to capture the global structure of a given differential equation. After the Borel sum of the series, the divergence can be regulated provided that the Borel summability is secured. Since this resummed function is a solution to the original differential equation by construction, it indeed provides the solution that has desired asymptotic expansions thanks to Watson's lemma \cite{watson1912vii,sokal1980,Aoki2019}.
The analysis on the property of the Borel-summed WKB series, therefore, enables us to study the global behavior of the corresponding solutions without losing advantages of the WKB method and without resorting to explicit functional forms of solutions.

One notable behavior of solutions that the standard WKB approximation cannot but the exact WKB can retrieve is so called \emph{Stokes phenomena}, in which asymptotic expansions change discontinuously while their appropriate linear combinations are continuous (and so is the full solution).
The exact WKB analysis handles Stokes phenomena on the complex coordinate plane, which occur across \emph{Stokes curves}, enabling analytical continuation of WKB solutions from one boundary to another {\it without approximations}.
It is then straightforward to impose boundary conditions at two asymptotics (incoming at the event horizon and outgoing at the spatial infinity in the case of black hole perturbations). 
The resulting conditions for QNMs can be precisely expressed by phase integrals of the odd-order WKB series, reducing the problem to computation of these phase integrals.

Applications of the exact WKB analysis to QNMs are rather limited in the literature. For instance, \cite{Imaizumi:2022dgj, Imaizumi:2022qbi} calculate QNMs of D-branes and M-branes using the exact WKB analysis.
However, to the best of our knowledge, its application to black hole QNMs has not yet been addressed, while the applicability of the exact WKB analysis in this context is only suggested in \cite{Hatsuda:2019eoj}.
This circumstance in the community has motivated our investigation in this paper.

The monodromy method \cite{Motl:2003cd, Andersson:2003fh} closely resembles the exact WKB analysis, as both utilize Stokes curves for analytic continuation. However, from the perspective of the exact WKB analysis, two key subtleties arise.
First, the monodromy method is strictly formulated for highly damped QNMs, and approximations are applied before analytic continuation. In contrast, the exact WKB analysis allows for exact analytic continuation, with approximations appearing only in the final stage of calculation of the phase integrals. Consequently, all difficulties are ultimately reduced to the evaluation of the aforementioned phase integrals.
Second, in the monodromy method, the radial or tortoise coordinate is analytically continued to a complex plane. Then, one has to choose a proper closed contour that leads to the correct QNM frequencies, which may not be obvious a priori.
In contrast, the exact WKB method allows us to construct a QNM eigenfunction on the real axis, which is advantageous when differential equations are originally defined on the real axis. One can systematically perform analytic continuation of a solution from one Stokes region to another along the contour that runs on the real axis. 
The path in the exact WKB analysis connects two asymptotic points where boundary conditions are imposed and is no longer a closed contour.

In this paper, we apply the exact WKB analysis to the calculation of black hole QNM frequencies.
We perform analytic continuation from the horizon to the spatial infinity along the positive real axis of the radial coordinate.
We reveal logarithmic spiral flows of Stokes curves that have been overlooked so far and branch cuts emerging from the horizon.
We pay special attention to incorporating these features into the calculation and reformulating the analytical derivation of QNM conditions using the exact WKB analysis.
With our methodology developed in this paper, distinct from the monodromy technique, we successfully derive the correct results in the case of solvable models and in the Regge-Wheeler equation, i.e.~a perturbed Schwarzschild black hole.

This paper is organized as follows.
In section \ref{sec:review_WKB}, we present a brief review of the exact WKB analysis.
In section \ref{sec:Eample}, we provide two known solvable models. 
First, we consider the harmonic oscillator as the simplest example to demonstrate how to analytically continue the exact WKB solutions. 
We then consider the Morse potential, a model simpler than but still similar to the Schwarzschild spacetime in terms of the Stokes geometry.
An important feature of the Stokes curves is their logarithmic spiral flow.
We perform analytic continuation of the WKB solutions by properly dealing with contributions due to the presence of such a flow.
In section \ref{sec:Schwarzschild}, we demonstrate our calculation in the case of Schwarzschild spacetime and recover the asymptotic values of highly damped QNM frequencies.
Section \ref{sec:conclusion} is devoted to the conclusion and discussion.
We complement the calculation of the residue of an infinite WKB series at regular singular points of differential equation in appendix \ref{app:regular_residue} and that of approximate results of phase integrals in appendix \ref{app:phase_int}.

\section{Review of exact WKB analysis}
\label{sec:review_WKB}
In this section, we briefly review the construction and usage of the exact WKB analysis and provide results without proof.
Interested readers can refer to~\cite{kawai2005algebraic,2014arXiv1401.7094I} for further mathematical details.

The exact WKB analysis is a powerful technique to analyze the global structure of solutions to linear second-order ordinary differential equations. It is thus well suited for finding QNMs, which is essentially an eigenvalue problem with given boundary conditions.
The formulation of the exact WKB is based on the resummation of an infinite series of a WKB solution using the so-called Borel sum. While a WKB series is often divergent, especially in physically interesting regimes, its Borel resummed counterpart is not and can give mathematically well-defined solutions, provided the Borel summability is ensured.

In the following subsections, we first construct WKB series solutions and their Borel summation. We then introduce the concept of the so-called Stokes curves, across which the behavior of asymptotic expansions of the solutions changes discontinuously. The effect of this change, called a Stokes phenomenon, is depicted by connection formulae, which we summarize toward the end of this section.

\subsection{Construction of WKB series solutions}

Our goal of this work is to compute the QNMs of perturbations around a spherically symmetric background spacetime. The perturbation equations of motion are thus second-order ordinary differential equations that are linear in the perturbed variables. 
The most general form of such an equation, without source terms, can be written as
\begin{align}
    \left[ \frac{\dd^2}{\dd x^2} + A(x) \, \frac{\dd}{\dd x} + B(x) \right] \varphi(x) = 0 \; ,
    \label{eq:general_secondorder_eq}
\end{align}
where $x$ and $\varphi(x)$ are the (dimensionless) radial coordinate and the perturbed variable, respectively, and $A(x)$ and $B(x)$ are functions of $x$ that are fixed once the background is specified. By changing variables by $\psi \equiv \exp \left[ \frac{1}{2} \int^x A(x') \, \dd x' \right] \varphi$, one can always recast \eqref{eq:general_secondorder_eq} into an equation of the Schr\"{o}dinger type. 
In order to construct the WKB series as a formal solution, we introduce a ``large'' parameter $\eta$
and rewrite the obtained equation as
\begin{align}
    \left[ -\frac{\dd^2}{\dd x^2}+\eta^2 Q(x, \eta) \right] \psi(x, \eta)=0\,,
    \label{eq:Schrodinger_type_eq}
\end{align}
which recovers eq.~(\ref{eq:general_secondorder_eq}) by taking $\eta = 1$ and changing variables from $\psi$ back to $\varphi$.
Here, $Q(x,\eta)$ is related to $A(x)$ and $B(x)$ by $Q(x,1) = -B + \frac{A^2}{4} + \frac{\partial_x A}{2}$. In order to apply the Borel sum, described below, $Q(x,\eta)$ should be a meromorphic function of $x$. Although $\eta$ is a parameter and not a variable so far, the Borel transformation we perform below treats $\eta$ as if it were a conjugate variable, and hence we write $\psi(x,\eta)$ as a function of both $x$ and $\eta$.

Let us note that the parameter $\eta$ is identified with $1/\hbar$ in the method of WKB approximation in quantum mechanics. However, in broader contexts where Schr\"{o}dinger-like equations are to be solved, as in our current interest of black hole perturbations, $\eta$ does not necessarily correspond to $1/\hbar$. Rather, $\eta$ is a formal parameter that is used to construct a formal power series as a solution to such a Schr\"{o}dinger-like equation. In this sense, $\eta$ is an auxiliary parameter to utilize the WKB analysis and should not be considered physically ``large.'' Notably, when working with dimensionless units in a given system, the original equation may lack suitable small or large parameters. In such a case, $\eta$ needs to be introduced manually into the potential $Q$ in a manner depending on one's physical interest. Once the full solution is obtained after the Borel summation of the formal WKB series, the solution is then applicable to any value of $\eta$, and in particular, taking $\eta = 1$ is a consistent choice, which recovers the solution to the original system before introducing $\eta$.

In transforming eq.~\eqref{eq:general_secondorder_eq} to eq.~\eqref{eq:Schrodinger_type_eq}, we may manually expand the function $Q(x,\eta)$ in terms of a power series in $\eta$, i.e.,
\begin{align}
    Q(x,\eta)
    = \sum_{j=0}^\infty \eta^{-j} Q_j(x)
    = Q_0(x)+\eta^{-1}Q_1(x)+\eta^{-2}Q_2(x) +\dots \; .
\end{align}
where $Q_{j} \, (j=0,1,\ldots)$ are meromorphic functions.
Notice that the way $\eta$ is implemented in $Q(x,\eta)$ is by no means unique. The choice needs to be made depending on each individual problem and on the quantities one is interested in. The strategy we take for choosing $\eta$ in the subsequent sections is such that asymptotic solutions around singular points of a given differential equation coincide with those of its leading-order WKB solutions. Asymptotic behaviors of WKB approximations may in general differ from those of the corresponding differential equation, and particularly equations with regular singular points need extra caution, see sections \ref{subsec:Morse_potential} and \ref{sec:Schwarzschild}. Since we impose boundary conditions at singular points of a differential equation, e.g.~the event horizon and the spatial infinity in the black hole spacetime, our choice of $\eta$ enables us to treat asymptotic behaviors in a consistent manner.

To construct a formal WKB series solution to eq.~\eqref{eq:Schrodinger_type_eq}, we take an ansatz for the form of $\psi$ as
\begin{align}\label{eq:WKBformal}
    \psi(x,\eta)= {\rm e}^{\int^x S(x',\eta) \, \dd x'}~,
\end{align}
where the function $S$ is the solution to a Riccati type equation, 
\begin{align}
    S^2+\frac{\dd S}{\dd x} = \eta^2 Q \; ,
    \label{eq:Riccati}
\end{align}
which is derived by feeding eq.~\eqref{eq:WKBformal} into eq.~\eqref{eq:Schrodinger_type_eq}.
We then expand $S$ as a power series in terms of $\eta$, that is,
\begin{align}
\label{eq:S_series}
    S(x,\eta)
    = \sum_{j=0}^\infty \eta^{1-j} S_{j-1}(x)
    = \eta S_{-1}(x)+\eta^0 S_{0}(x)+\eta^{-1} S_{1}(x)+\eta^{-2} S_{2}(x)+\dots \; .
\end{align}
Plugging this into \eqref{eq:Riccati} and equating the left-hand and right-hand sides order by order in $\eta$, 
we obtain the recursion relation,
\begin{subequations}
\label{eq:recursion}
\begin{align}
    &S_{-1}^2=Q_0 \; , \qquad
    S_0 = - \frac{1}{2 S_{-1}} \left( \frac{\dd S_{-1}}{\dd x} - Q_1 \right) \; , \\
    & S_{j+1} = - \frac{1}{2 S_{-1}} \left( \frac{\dd S_j}{\dd x} + \; \sum_{\mathclap{\substack{n+m = j \\ n \ge 0 , m \ge 0}}} S_n S_m - Q_{j+2} \right) \; , \quad 
    \left( j= 0,1,2, \dots \right) \; .
    \label{eq:S_j}
\end{align}
\end{subequations}
All the terms in the power series can be fixed recursively, and moreover, each $S_{j}$ is obtained solely in terms of lower orders and known functions $Q_j$. Notably, one can find all the $S_j$ algebraically and no longer needs to solve the differential equation \eqref{eq:Riccati} to find the solution. This series expansion in $\eta$ is thus an example of singular perturbations. In this respect, it is not surprising that the na\"{i}ve limit $\eta \to \infty$ may not exhibit a proper behavior of the original differential equation, and we need further treatment to obtain meaningful interpretations of the series solutions.

To continue constructing the WKB series, we observe from eq.~\eqref{eq:recursion} that the leading term has two roots: $S_{-1}^{\pm}=\pm\sqrt{Q_0}$. Once the sign of $S^\pm_{-1}$ is chosen, all the subleading terms in $S$ are computed uniquely through the recursion relation \eqref{eq:recursion}. Denoting the respective series by $S^\pm$, we obtain two independent solutions, i.e.,
\begin{align}
\label{eq:def_Sjpm}
    S^{\pm}=\sum_{j=-1}^{\infty}\eta^{-j}S^{\pm}_j\,,
\end{align}
where the superscript $\pm$ corresponds to the branch of $S_{-1}^{\pm}$. 
Hereafter, we assume that only even order terms in $Q$ may take nonzero values and all the odd terms are absent, that is $Q_{2j-1} = 0$ for any positive integer $j$.%
\footnote{As a matter of fact, this assumption is not strictly necessary. By finding $S^\pm$ with respective $S_{-1} = \pm \sqrt{Q_0}$ and the recursion relation \eqref{eq:recursion}, one can again define $S_{\rm odd} = (S^+ - S^-)/2$ and $S_{\rm even} = (S^+ + S^-)/2$. Then the relation \eqref{eq:Soddeven_relation} and the formal solutions \eqref{eq:WKB_solutions} still follow. However, in the presence of nonzero $Q_{2j-1}$, relations such as $S_j^+ = (-1)^j S_j^-$ would no longer hold. To reduce the complexity, we here make an assumption $Q_{2j-1} =0$, which is compatible with our following analysis.}
Under this condition, it is evident from the recursion relation \eqref{eq:recursion} that the sign of $S_{-1}^\pm$ is inherited to every odd-order term of $S_j$, while all the even-order terms are independent of the choice. This fact can be concisely expressed as the relation $S_j^{-} = (-1)^j S_j^{+}$ for $j \geq -1$.
The solutions $S^{\pm}$ can then be written in a compact form~%
\footnote{To justify changing the order of summation in eq.~\eqref{eq:S_series} to separate the odd- and even-order terms as in eq.~\eqref{eq:Spmoddeven}, absolute convergence of the original series is in principle required, if one cares about actual values of the series. However, we are merely constructing the formal WKB series here, and this re-ordering is nothing but part of the formal mathematical procedure.}
\begin{align}\label{eq:Spmoddeven}
    S^{\pm}=\pm S_{\text{odd}}+S_{\text{even}}~,
\end{align}
where we define
\begin{align}
    &S_{\text{odd}}:=\frac{1}{2}(S^+ - S^-)=\sum_{j=0}^{\infty} \eta^{1-2j}S_{2j-1}~,\\
    &S_{\text{even}}:=\frac{1}{2}(S^+ + S^-)=\sum_{j=0}^{\infty} \eta^{-2j}S_{2j}~.
\end{align}
In this expression, we denote $S_j \equiv S_j^+ = (-1)^j S_j^-$.
Plugging eq.~\eqref{eq:Spmoddeven} into the Riccati equation \eqref{eq:Riccati} and formally separating the odd and even orders, the equality of the odd-order part implies that $S_{\text{even}}$ and  $S_{\text{odd}}$ are related through
\begin{align}
\label{eq:Soddeven_relation}
    S_{\text{even}}=-\frac{1}{2}\frac{\dd}{\dd x} \log S_{\text{odd}}~.
\end{align}
Contributions from $Q$ do not appear in this equation, since $Q_{2j-1} = 0$ for all positive integer $j$ is assumed.
Substituting eq.~\eqref{eq:Spmoddeven} into eq.~\eqref{eq:WKBformal} together with the relation \eqref{eq:Soddeven_relation}, the formal WKB solutions are given by
\begin{align}
    \psi^{\text{WKB}}_{\pm}(x,\eta)=\frac{1}{\sqrt{S_{\text{odd}}}}\exp\bigg(\pm\int^x_{x_i} S_{\text{odd}}(x') \, \dd x' \bigg)~,
    \label{eq:WKB_solutions}
\end{align}
where $x_i$ is an arbitrary reference point.
The two functions $\psi^{\rm WKB}_\pm$ are independent solutions, at least order by order in $\eta$, to the differential equation \eqref{eq:Schrodinger_type_eq}  and span the basis of its solution space, provided that the ansatz \eqref{eq:WKBformal} is valid.

In the following formulation of the \textit{exact} WKB analysis, it is convenient to take the reference point $x_i$ as one of the zeros of $Q_0$, called {\it turning points}. 
Hereafter, we assume that all the turning points are simple, i.e. the order of $x_i$ is 1, or $\dd Q_0 / \dd x \, \vert_{x = x_i} \ne 0$, which is true for the black hole spacetimes we consider in this work. 
Since $Q$ is a meromorphic function of $x$, and so is $Q_0$, $S_{-1} = \sqrt{Q_0}$ is two-valued. As the recursion formula \eqref{eq:recursion} shows, each odd-order term in $S$ inherits the multi-valuedness, while even-order terms are all single-valued. The WKB solutions $\psi_\pm^{\rm WKB}$ consist of $S_{\rm odd}$ alone, and thus we must take good care of the branches of $S_{\rm odd}$. The multi-valuedness of $S_{\rm odd}$ originates from that of $S_{-1} = \sqrt{Q_0}$, and the singular behavior in $S_{2j-1}$ around each simple turning point $x_i$ is only of the form $(x-x_i)^{-m/2}$ with $m$ an odd integer \cite{kawai2005algebraic}. 
For this reason,  the integral \eqref{eq:WKB_solutions} can be defined by half of the contour integral,
\begin{align}
\label{eq:path_integral}
    \int^x_{x_i} S_{\text{odd}} \, \dd x'
    =\frac{1}{2}\int_{\gamma_x}S_{\text{odd}} \, \dd x'~,
\end{align}
along a path $\gamma_x$, which starts from the counterpart point of $x$ on the second Riemann sheet, goes around the turning point $x_i$ while passing a branch cut associated with $x_i$, and runs through $x$ on the first Riemann sheet (see figure \ref{fig:normalization}).%
\footnote{Since the sign of $S_{\text{odd}}$ flips on the second sheet of the Riemann surface, the contributions from the upper and lower contours in figure~\ref{fig:normalization} are the same. For some basics about Riemann surfaces, we refer readers to, e.g.~\cite{Needham:579291,FarkasKra1992,Donaldson2011}.}
The aforementioned singularities $(x-x_i)^{-m/2}$ are canceled in $\gamma_x$, and the integral in eq.~\eqref{eq:WKB_solutions} becomes finite.
Note that double poles in $Q_0$ correspond to regular singular points of the original differential equation.
Note also that, while the exponential part of $\psi_\pm^{\rm WKB}$ in eq.~\eqref{eq:WKB_solutions} is only two-valued, $\psi_\pm^{\rm WKB}$ as a whole contains additional branches due to the presence of $S_{\rm odd}^{-1/2} \sim (\eta^2 Q_0)^{-1/4}$.
\begin{figure}
    \begin{tabular}{cc}
        \begin{minipage}[b]{0.48\linewidth}
        \centering
        \includegraphics[keepaspectratio, scale=0.28]{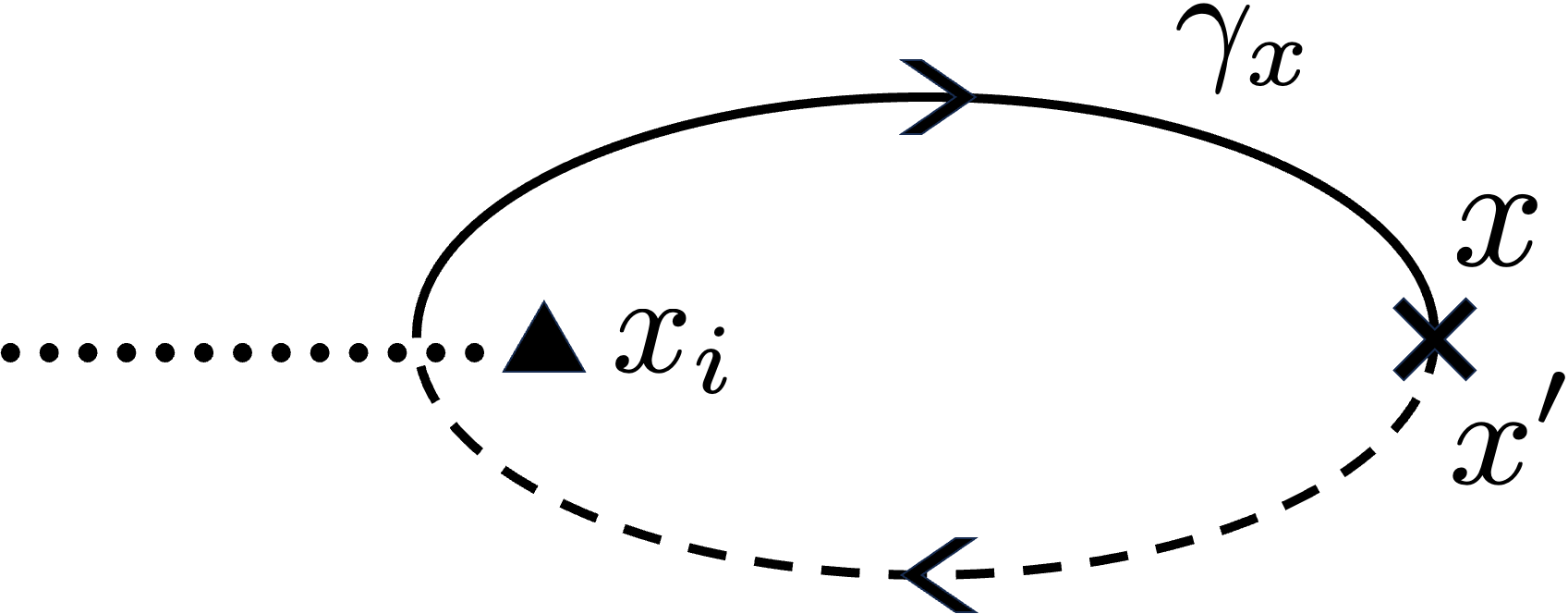}
        \end{minipage}&
        \begin{minipage}[b]{0.48\linewidth}
        \centering
        \includegraphics[keepaspectratio, scale=0.28]{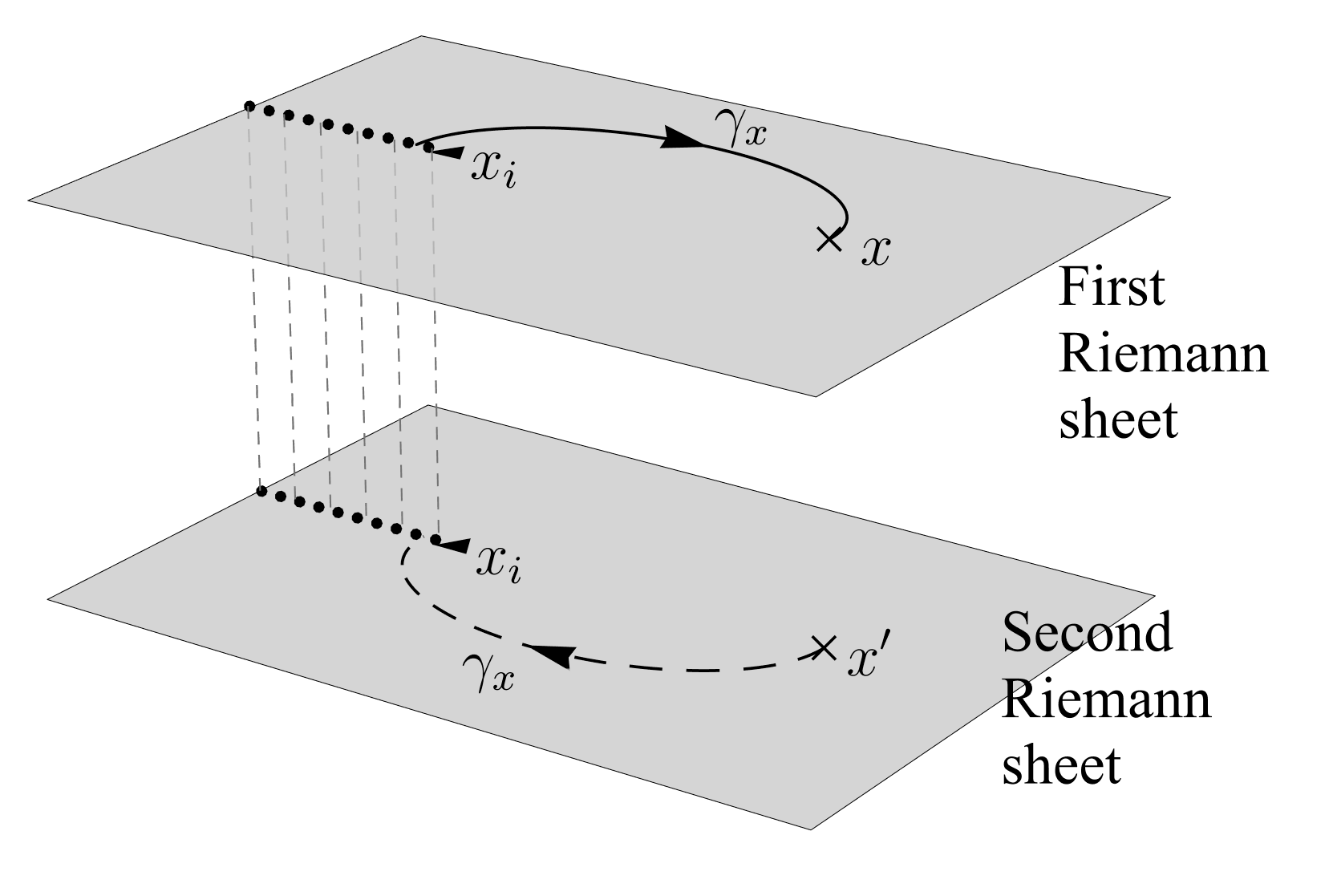}
        \end{minipage}
    \end{tabular}
    
    \caption{(Left) The contour $\gamma_x$. The black triangle $x_i$ is a turning point, and the dotted line is the branch cut of $\sqrt{Q_0(x)}$.
    The point $x$ $(x')$ is on the first (second) Riemann sheet.
    We use dotted lines for branch cuts throughout this paper.
    (Right) The 3D representation of the Riemann surface and contour $\gamma_x$. The first and second Riemann sheets are connected by branch cuts from the tuning point $x_i$ (dotted line).}
    \label{fig:normalization}
\end{figure}

The two formal WKB solutions $\psi_\pm^{\rm WKB}$ in eq.~\eqref{eq:WKB_solutions} would serve as the true independent solutions to the differential equation \eqref{eq:Schrodinger_type_eq} if the series in $\eta$ were convergent. However, this is not the case in most, if not all, physically interesting systems. For example, even in the simplest but nontrivial case with $Q=Q_0=x$, the series \eqref{eq:S_series} computed by the recursion relation \eqref{eq:recursion} is divergent with $0$ convergence radius for all complex values of $x$ (or equivalently $\eta$ in this case) except for exactly $x=0$. This does not necessarily come as a surprise, as the WKB series is a singular perturbation expansion, as mentioned below eq.~\eqref{eq:recursion}. As long as the terms $Q_j$ are suitably chosen, the WKB solutions may grasp correct asymptotic behaviors with arbitrary precision. However, in general, they cannot cover the entire domain of a solution and require an additional technique for a global analysis of the system. In particular, asymptotic expansions may change over a (complex) domain of $x$, called \textit{Stokes phenomenon}, which occurs when $x$ is continuously changed to cross so-called \textit{a Stokes curve}. This behavior cannot be captured by a divergent WKB series. To resolve the issue, we employ the \textit{Borel summation} to re-sum the series to a convergent one and to give meaningful interpretations of its global behaviors in the following subsections.

\subsection{Borel summation}
\label{subsec:BolreSum_WKB}

The formal WKB solutions \eqref{eq:WKB_solutions} can be further expanded as a power series in $\eta^{-1}$. The leading order in $S_{\rm odd}(x,\eta)$ is of the order $\eta$ and thus stays in the exponent in $\psi_\pm^{\rm WKB}$; the remaining terms in $S_{\rm odd}$ are suppressed by the formally large parameter $\eta$ and can be expanded, giving the expression,%
\footnote{The factor $-1/2$ on the power of $\eta$ comes from $1/\sqrt{S_{\text{odd}}}$ in \eqref{eq:WKB_solutions}.}
\begin{align}
    \psi^{\text{WKB}}_{\pm}(x,\eta)=\ee^{\pm\eta y_0(x)}
    \sum_{n=0}^\infty\psi_{\pm,n}(x) \, \eta^{-n-1/2}~,
    \label{eq:WKB_solutions_expansion}
\end{align}
where
\begin{align}
    y_0(x)=\int_{x_i}^x\sqrt{Q_0(x')} \, \dd x'~.
\end{align}
Generally, the coefficients $\psi_{\pm,n}$ diverge for large $n$, even with a vanishing convergence radius, in physically interesting circumstances,
and the WKB solutions themselves are ill-defined.
The strength of the divergence often falls into the so-called Gevrey-$1$ category, namely $\vert \psi_{\pm,n} \vert \le C A^n n!$ for large $n$ with some positive constants $A$ and $C$. In the following formulation, we always assume that the WKB series $\psi_{\pm,n}$ be Gevrey-$1$.

The idea of exact WKB analysis is to apply {\it the Borel summation} to this series.
The Borel summation is a resummation method that gives an analytical meaning to the Gevrey-1 formal series.
It consists of the {\it Borel transformation}
and the Laplace transform.
The Borel sum $\mathcal{S}[\psi_\pm](x,\eta)$ of the original WKB series $\psi_\pm^{\rm WKB}(x,\eta)$ is defined as the Laplace transform of the Borel-transformed function $\mathcal{B}[\psi_\pm](x,y)$ associated with $\psi_\pm^{\rm WKB}(x,\eta)$, i.e.,
\begin{align}
    \mathcal{S}[\psi_{\pm}](x,\eta) := \int_{\mp y_0(x)}^\infty {\rm e}^{-\eta y}\mathcal{B}[\psi_{\pm}](x,y) \, \dd y~.
    \label{eq:BorelSum}
\end{align}
where the integration path is $y\in\{\mp y_0(x)+t;\, t\geq 0\}$~%
\footnote{The integration path is parallel to the real axis when $\eta$ is (positive) real. In the case of complex $\eta = \vert \eta \vert {\rm e}^{i \theta}$, the path is suitably deformed such that the integral starts from $y_0(x)$ and runs to infinity with the angle $-\theta$.}
and
\begin{align}
    \mathcal{B}[\psi_{\pm}](x,y) := \sum_{n=0}^\infty \frac{\psi_{\pm,n}(x)}{\Gamma(n+1/2)} \left[ y\pm y_0(x) \right]^{n-1/2},
    \label{eq:BorelTransformation}
\end{align}
is the Borel transformation of $\psi_\pm^{\rm WKB}$.
If $\sqrt{y \pm y_0} \, \mathcal{B}[\psi_\pm]$ converges at $y=\mp y_0$ and can be analytically continued in the region that contains the integration path, and if the Laplace integral in eq.~\eqref{eq:BorelSum} yields a finite and definite value for sufficiently large $\eta$, then the original WKB solutions are said to be {\it Borel summable}.
Under the condition of Borel summability, that the original WKB approximation coincides with the asymptotic expansion of the Borel-summed solution is ensured by the so-called Watson's lemma \cite{watson1912vii,sokal1980,Aoki2019}.

The Borel summability depends on the variable $x$, and it is known that the Borel summability may not hold on the {\it Stokes curves}, which are defined on the complex $x$ plane by the following criterion
\begin{align}
    {\rm Im}\big[ y_0(x) \big] = {\rm Im} \int_{x_i}^x \sqrt{Q_0(x')} \, \dd x'=0.
    \label{eq:SL_definition}
\end{align}
When all turning points of $Q_0$ are simple, and the order of poles is higher than or equal to two,%
\footnote{A simple pole has a roll of turning point and gives different connection formulae \cite{Koike:2000}.}
three Stokes curves emerge from each turning point and flow into singular points (poles) or turning points  (see figures~ \ref{fig:SL_harmonic}, \ref{fig:SL_schematic_Morse}, \ref{fig:StokesLines_qnm} and \ref{fig:SL_schematic_Schwarzschild} for example).
When the Stokes geometry is \textit{non-degenerate}, that is, when 
none of the Stokes curves connects two turning points or forms a loop emerging from and going into the same turning point, the WKB solutions are Borel summable in each {\it Stokes region} surrounded by Stokes curves.%
\footnote{This claim is noted in e.g.~\cite{Aoki2019} and expected to be proven in \cite{Koike:inprep}, which has not been completed due to the passing of Koike in 2018.}

\subsection{Connection formula}
\label{subsec:StokesLines_Connection}
One of the WKB solutions $\psi_\pm^{\rm WKB}$ is not Borel summable on a Stokes curve, due to the fact that its integration path in eq.~\eqref{eq:BorelSum} hits one of the singular points $y=\pm y_0(x)$ and the corresponding branch cut on the complex $y$ plane that is originated from the fractional power in eq.~\eqref{eq:BorelTransformation}.%
\footnote{This singular point is at $y=+y_0$ for $\mathcal{B}[\psi_+]$ and $y=-y_0$ for $\mathcal{B}[\psi_-]$, and thus it is not obvious from the expression \eqref{eq:BorelTransformation} that $y=-y_0$ (respectively $y=+y_0$) forms a singularity on the $y$ plane for $\mathcal{B}[\psi_+]$ (respectively $\mathcal{B}[\psi_-]$). However, $\mathcal{B}[\psi_\pm]$ should solve the (partial) differential equation that is derived by replacing $\eta$ by $\partial / \partial y$ in the original Schr\"{o}dinger-type equation \eqref{eq:Schrodinger_type_eq}. Provided that $Q(x,\eta)$ only contains even-order terms in $\eta$, which we have been assuming, the resultant differential equation has symmetry under $y \to -y$. Hence each of $\mathcal{B}[\psi_\pm]$ carries both of the singular points at $y=\pm y_0$.}
The WKB solutions are Borel summable within each Stokes region, provided the Stokes geometry is non-degenerate. 
When the variable $x$ is moved to cross a Stokes curve, the Stokes phenomenon occurs, and asymptotic behaviors of the solutions change. Since they form asymptotic expansions, the WKB solutions constructed in different Stokes regions, in general, differ from each other and are related through their linear combinations because the differential equation \eqref{eq:Schrodinger_type_eq} is linear.
This means that, when the WKB solutions constructed in one Stokes region are analytically continued across a Stokes curve, the integration path in the Borel sum \eqref{eq:BorelSum} needs to be deformed, and one finds different parts of the deformed path correspond to the Borel sum of the WKB solutions constructed in the other region.

In order to summarize how the Borel-summed, exact WKB solutions in different Stokes regions are connected,
let us consider a Stokes curve $\Gamma$ emanating from a turning point $x_i$.
If
\begin{align}
    \label{eq:dominance_positive}
    \text{Re}\, \int^x_{x_i} \sqrt{Q_0(x)} \, \dd x > 0 \quad \text{along $\Gamma$}~,
\end{align}
the exact WKB solution $\mathcal{S}[\psi_+]$ is said to be dominant over $\mathcal{S}[\psi_-]$ on $\Gamma$, that is, the latter is exponentially suppressed compared to the former.
Else if
\begin{align}
    \label{eq:dominance_negative}
    \text{Re}\, \int^x_{x_i} \sqrt{Q_0(x)} \, \dd x < 0 \quad \text{along $\Gamma$}~,
\end{align}
the other solution $\mathcal{S}[\psi_-]$ is dominant over $\mathcal{S}[\psi_+]$ on $\Gamma$.%
\footnote{From the Cauchy-Riemann equations, it can be shown that the real part of the integral in eq.~\eqref{eq:dominance_positive} (eq.~\eqref{eq:dominance_negative}) is monotonically increasing (decreasing), provided that the Stokes curve $\Gamma$ does not connect multiple turning points.}
Now we consider analytically continuing the solutions in two regions I and II separated by a Stokes curve $\Gamma$. 
Only the dominant solution, of which the Borel transformation runs into the singular point at $y=\pm y_0$ along its integration path, experiences the Stokes phenomenon, while the subdominant ones in the two regions coincide. As a result,
the Borel-summed WKB solutions in region I are related to the solution in region II by the connection formulae
\begin{align}
\label{eq:connection_formulae}
\begin{cases}
    \mathcal{S}[\psi_+^{\text{I}}] = \mathcal{S}[\psi_+^{\text{II}}] + i\mathcal{S}[\psi_-^{\text{II}}]~, \quad \mathcal{S}[\psi_-^{\text{I}}] = \mathcal{S}[\psi_-^{\text{II}}]~, \qquad & (\text{$\psi_+$ is dominant on $\Gamma$})~, \vspace{3mm} \\ 
    \mathcal{S}[\psi_-^{\text{I}}] = \mathcal{S}[\psi_-^{\text{II}}] + i\mathcal{S}[\psi_+^{\text{II}}]~, \quad \mathcal{S}[\psi_+^{\text{I}}] = \mathcal{S}[\psi_+^{\text{II}}]~, \qquad & (\text{$\psi_-$ is dominant on $\Gamma$})~,
\end{cases}
\end{align}
when the analytic continuation is performed by a path crossing the Stokes curve $\Gamma$ counterclockwise around the associated turning point.
The direction of crossing is determined locally in the vicinity of the turning point in relation to the Stokes curved attached to it.
The connection formulae can be compactly written in the $2\times 2$ matrix form
\begin{align}
\label{eq:connection_matrix}
    \begin{pmatrix}
        \psi_+\\
        \psi_-
    \end{pmatrix}^{\text{I}}
    = V_{\pm}
    \begin{pmatrix}
        \psi_+\\
        \psi_-
    \end{pmatrix}^{\text{II}}~,
\end{align}
with
\begin{align}\label{eq:connectionS}
            V_+ =
            \begin{pmatrix}
                1 & i\\
                0 & 1
            \end{pmatrix}~, 
            \qquad 
            V_- =
            \begin{pmatrix}
                1 & 0\\
                i & 1
            \end{pmatrix}~.
    \end{align}
The $+$ ($-$) subscript corresponds to the case when the solution $\psi_{+}$ ($\psi_{-}$) is dominant on the Stokes curve $\Gamma$.
For a clockwise crossing, the sign in front of $i$ is flipped.
This change of coefficients reflects the effect of the Stokes phenomenon.
Note that the coefficient $i$ is given exactly without relying on any approximation, provided the Borel summability is assured and the WKB integral is normalized at the turning point $x_i$ the Stokes curve is attached to as in eq.~\eqref{eq:WKB_solutions}. Therefore one can freely perform analytic continuation at any point on the complex $x$ plane to span all the regions.

\begin{figure}
    \centering
    \includegraphics[width=0.35\linewidth]{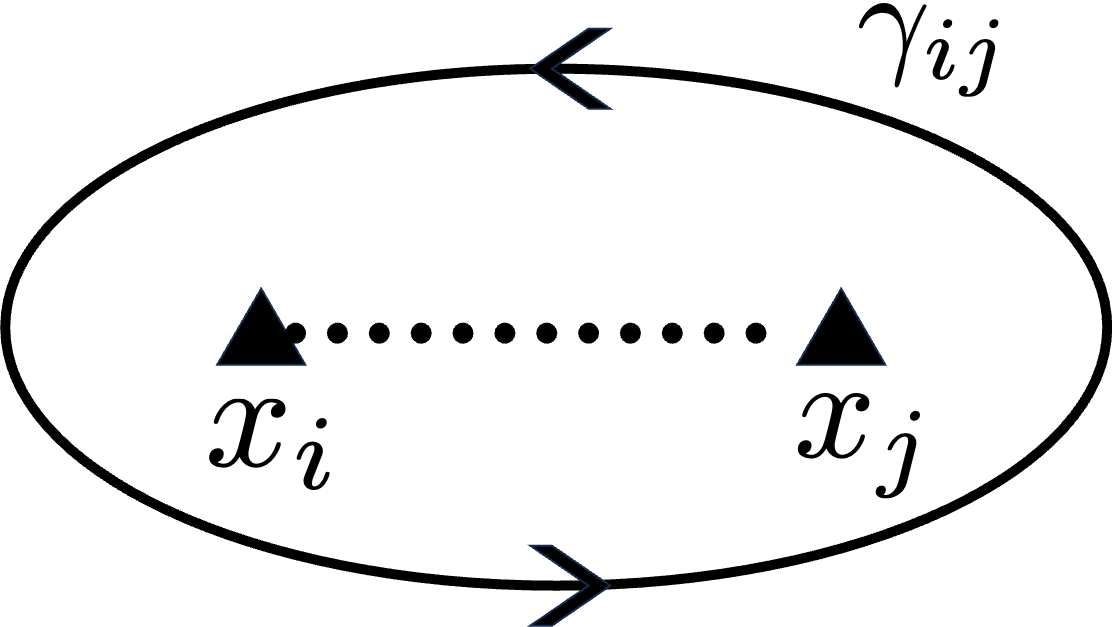}
    \caption{The contour $\gamma_{ij}$. The black triangles $x_{i,j}$ are turning points, and the dotted line is the branch cut of $\sqrt{Q_0(x)}$. The orientation is defined by the definition of $\gamma_{x}$ in figure~\ref{fig:normalization}.
    }
    \label{fig:gamma_ij}
\end{figure}
\begin{figure}
    \centering
    \includegraphics[width=0.35\linewidth]{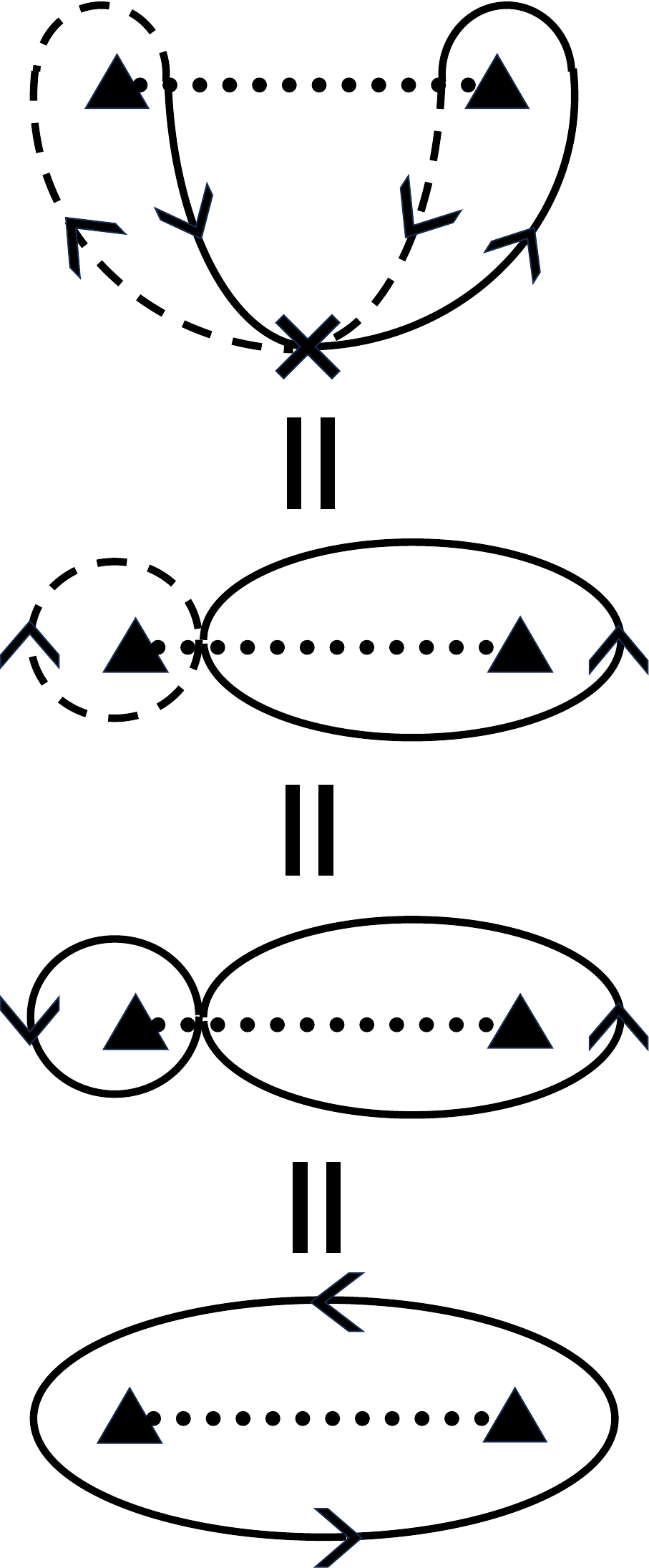}
    \caption{
    The deformation of the path to $\gamma_{ij}$.
    Firstly, As in figure \ref{fig:normalization}, one can redraw the original path $x_i \to x \to x_j$ to look like the top path in this figure. 
    In the first equality, one merely deforms the path.
    In the second equality, we can replace the path in the second Riemann sheet (dashed curve) to the path in the first Riemann sheet (solid curve).
    At this point, because of the branch of the square root function, there is a phase difference of $\pi$, resulting in the path being traversed in the opposite direction. 
    Finally, by deforming the path, we obtain the path as in the figure \ref{fig:gamma_ij}.
    }
    \label{fig:PuzzleRing}
\end{figure}

In general, multiple turning points exist, requiring us to cross Stokes curves originating from different turning points. In such a situation, one needs to change the reference turning points of the integral in eq.~\eqref{eq:WKB_solutions} as one encounters a new Stokes curve,
since the connection formulae \eqref{eq:connection_formulae} apply to the WKB solutions \eqref{eq:WKB_solutions} that are normalized at the specific turning point attached to the Stokes curve that is being crossed. 
In changing turning points from $x_i$ to $x_j$, the WKB solutions acquire a (complex) phase
\begin{align}
    \ee^{\pm \int_{x_i}^{x_j} S_{\rm odd} \, \dd x}~,
\end{align}
which is evident from the definition~\eqref{eq:WKB_solutions}. Thus, we multiply the matrix
\begin{align}\label{eq:connectionT}
    \Gamma_{ij} =
    \begin{pmatrix}
        u_{ij} & 0\\
        0 & u_{ij}^{-1}
    \end{pmatrix}~, \qquad
    u_{ij} = \exp\bigg(\int_{x_i}^{x_j} S_{\text{odd}}\, \dd x\bigg)~,
\end{align}
whenever we need to change the reference points.
The integral in $u_{ij}$ should be evaluated according to the path described in eq.~\eqref{eq:path_integral} and figure \ref{fig:normalization}. This time, however, the integral path starts from a turning point $x_i$, runs to a point $x$ in the Stokes region where the switching of turning points takes place, and then to another turning point $x_j$, which, combined, forms a closed contour. Thus, $u_{ij}$ is evaluated as
\begin{align}
\label{eq:uij}
    u_{ij} = \exp\left( \frac{1}{2} \oint_{\gamma_{ij}} S_{\rm odd} \, \dd x \right) \;
\end{align}
where the closed contour $\gamma_{ij}$ is schematically shown in figure \ref{fig:gamma_ij}. The direction of the contour to be evaluated is fixed by inspecting the original path $x_i \to x \to x_j$.
In figure~\ref{fig:PuzzleRing}, we show an example of such a path deformation to $\gamma_{ij}$.
Note that, from the definition, $u_{21} = u_{12}^{-1}$ and $\Gamma_{21} = \Gamma_{12}^{-1}$.

The solutions also acquire a nontrivial phase factor when one needs to cross a branch cut emanating from a
singular point in $Q_0(x)$.%
\footnote{In principle, crossing the branch cut attached to a turning point can yield a phase factor as well, but in such a case the solution under consideration only enters in another Riemann sheet, and thus the phase is trivial.} 
This branch cut is not necessarily the ones attached to turning points and hence does not always invoke branches for the exponent $\exp \big( \int_{x_i}^x S_{\rm odd} \, \dd x \big)$ in the WKB solutions. Instead, it is associated with the factor $1/\sqrt{S_{\rm odd}}$. In the case of regular singular points, i.e.~$Q_0 \propto (x- b_i)^{-2}$ around the $i$-th singular point $b_i$, $1/\sqrt{S_{\rm odd}} \sim (\eta^2 Q_0)^{-1/4}$ is two-valued around $b_i$, while the exponential factor is single-valued. When crossing the branch cut from $b_i$, one goes to the second Riemann sheet; when coming back to the same point on the first sheet, $1/\sqrt{S_{\rm odd}}$ changes its phase by $\pi$, and the exponential factor $\exp \big( \int_{x_i}^x S_{\rm odd} \, \dd x \big) = \exp \big( \frac{1}{2} \oint_{\gamma_x} S_{\rm odd} \, \dd x \big)$ picks up the residue at $b_i$.
In total, one obtains a phase factor
\begin{align}
\label{eq:crossing_branchcut}
    \nu_i^{\pm} &= \exp\left[ i\pi \, \big( 1\pm 2 \underset{x=b_i}{\text{Res}}\,S_\text{odd} \big) \right]~,
\end{align}
when $\psi_{\pm}$, respectively, cross a branch cut associated with a regular singular point. In terms of the $2\times 2$ matrix form, this is 
    \begin{align}
        D_i &= 
        \begin{pmatrix}
            \nu_i^+ & 0\\
            0 & \nu_i^-
        \end{pmatrix}~.
        \label{eq:MatrixD}
    \end{align}
Notice that the connection matrices $V_\pm$ contain off-diagonal components and mix the $\pm$ solutions, while $\Gamma_{ij}$ and $D_i$ only rescale each of $\psi_\pm$ individually. Yet, it is essential to properly include $\Gamma_{ij}$ and $D_i$ along with $V_\pm$ in order to obtain correct solutions to the original differential equation. 

The recipes to analytically continue the exact WKB solutions in two Stokes regions $A$ and $B$ can be summarized as follows. First, one draws Stokes curves and branch cuts on the complex $x$ plane and determines the dominance relation of the WKB solutions. Then, consider the contour connecting region $A$ to $B$. 
We start with the vector of the Borel-summed WKB solutions constructed in region $A$, i.e.,%
\footnote{Strictly speaking, $\psi_\pm$ should be written as $\mathcal{S}[\psi_\pm]$, but we abuse the notation here for concise expressions. Hereafter, the Borel summation is assumed whenever no ambiguity is expected.}
\begin{align}
\begin{pmatrix}
    \psi_+\\ \psi_-
\end{pmatrix}^{(A)} \; ,
\end{align}
and express it as the same vector in region $B$ multiplied from the left by the matrix $V_{\pm}$ (see eq.~\eqref{eq:connectionS}), i.e.,
\begin{align}
    \label{eq:connection_summary}
    \begin{pmatrix}
    \psi_+\\ \psi_-
    \end{pmatrix}^{(A)}
    = 
    V_\pm
    \begin{pmatrix}
    \psi_+\\ \psi_-
    \end{pmatrix}^{(B)} \; ,
\end{align}
when the contour crosses a Stokes curve anti-clockwise (multiplied by $V_{\pm}^{-1}$ when crossing clockwise). 
When crossing a Stokes curve originating from $x_i$, followed by another one originating from $x_j$, multiply $\Gamma_{ij}$ in eq.~\eqref{eq:connectionT} before crossing the latter. Finally, when the contour encounters a branch cut emanating from a regular singular point $b_i$, multiply $D_i$ shown in eq.~\eqref{eq:MatrixD}. In the next section, we show how to apply these recipes to the computation of the eigenvalues of two-point boundary value problems.

\section{Example of exact WKB analysis}
\label{sec:Eample}

In this section, we demonstrate how to utilize the exact WKB analysis in simple examples, before investigating our main interest of QNMs of black-hole perturbations in the following section.
In order to show the core of the analysis, we apply the exact WKB method to two well-known problems in quantum mechanics: a harmonic oscillator (section \ref{subsec:Harmonic_oscillator}) and the Morse potential (section \ref{subsec:Morse_potential}). 
Readers who are interested in the final result of black hole QNMs may skip this section and jump to section \ref{sec:Schwarzschild}.

In section~\ref{subsec:Harmonic_oscillator}, we consider a harmonic oscillator to illustrate how the exact WKB analysis can be applied to an eigenvalue problem.
In section~\ref{subsec:Morse_potential}, we introduce the Morse potential, which exhibits a singularity structure similar to that of the Schwarzschild black hole.
In this example, we demonstrate the existence of a logarithmic spiral of Stokes curves around a regular singular point, as well as the necessity of deriving a connection formula. To our knowledge, this is a novel application of exact WKB analysis to an eigenvalue problem with the Morse potential.\footnote{After posting our paper, another paper \cite{Morikawa:2025grx} on a similar calculation for the Morse potential appeared on arXiv.}

\subsection{Harmonic oscillator}
\label{subsec:Harmonic_oscillator}

Let us first consider the simplest problem, finding the eigenvalues of a harmonic oscillator.
The potential of the current problem is given by
\begin{align}\label{eq:harmonicpot}
    Q(x,\eta) = Q_0(x) =\frac{1}{4} \, x^2-E\,,
\end{align}
where $E$ is a constant. The eigenvalue $E$ is determined by imposing the wave function to decay as $x \to \pm \infty$.
From the standard quantum mechanics textbooks, we know that the eigenvalues are~\footnote{In quantum mechanics, the $\eta^{-1}$ should be identified as the Planck constant, i.e., $\eta^{-1}=\hbar$.}
\begin{align}\label{eq:eigen_harmonic}
    E &=\eta^{-1}\left(n+ \frac{1}{2}\right), & n&=0,1,2,\cdots~.
\end{align}
Let us see how this result can be obtained from the exact WKB analysis.

The first step of the exact WKB analysis is to draw the Stokes curves according to their definition~\eqref{eq:SL_definition}. 
Those of the potential \eqref{eq:harmonicpot} are depicted in figure \ref{fig:SL_harmonic}.%
\footnote{It is implicitly assumed that $E$ is real and positive. While this is obvious from the physics point of view, the exact WKB analysis also derives this result. Even if one starts from the Stokes geometry with a complex $E$, the requirement to satisfy the boundary conditions results in essentially the same one as below, i.e.~$u_{21}^2 = -1$ with $u_{21}$ given in eq.~\eqref{eq:u21_harmonic}, giving real and positive eigenvalues. In particular, for a real and negative $E$, the path of the analytic continuation along the real axis of $x$ crosses no Stokes curves, and therefore no negative eigenvalues are present for the given boundary conditions.}
In this regard, the domain of $x$ is extended to the entire complex plane.
To determine the branch, we choose $\sqrt{Q(x)}>0$ for $x>2\sqrt{E}$ and take the branch cut between two turning points, $x_1=-2\sqrt{E}$ and $x_2=2\sqrt{E}$. 
With this choice, the asymptotic form of the WKB solutions as $x\to +\infty$ is 
\begin{align}
    \psi^{\rm WKB}_{\pm}(x,\eta) \to \left( \frac{4}{\eta^2 x^2} \right)^{1/4} \exp\left(\pm \frac{\eta x^2}{4} \right)~.
\end{align}
Thus, $\psi_+$ is dominant on the Stokes curve $\{ x \, \vert \, x\in \mathbb{R}, x > x_2\}$. Respecting the symmetry of the system under $x \to -x$, the solution $\psi_+$ is also dominant on the Stokes curve $\{ x \, \vert \, x\in \mathbb{R}, x < x_1\}$. Other cases can be analyzed similarly (see figure \ref{fig:SL_harmonic}).%
\footnote{Once the dominance relation is known for one Stokes curve coming out of a turning point, the easiest way to find the relation for other Stokes curves from the same turning point is by the fact that the dominance relation alternates between adjacent Stokes curves. Namely, if $\psi_+$ is dominant on a Stokes curve, then $\psi_-$ is dominant on the (two) Stokes curves next to it. This alternating relation continues through Riemann sheets when a branch cut is crossed.}
\begin{figure}
    \centering
    \includegraphics[width=0.5\linewidth]{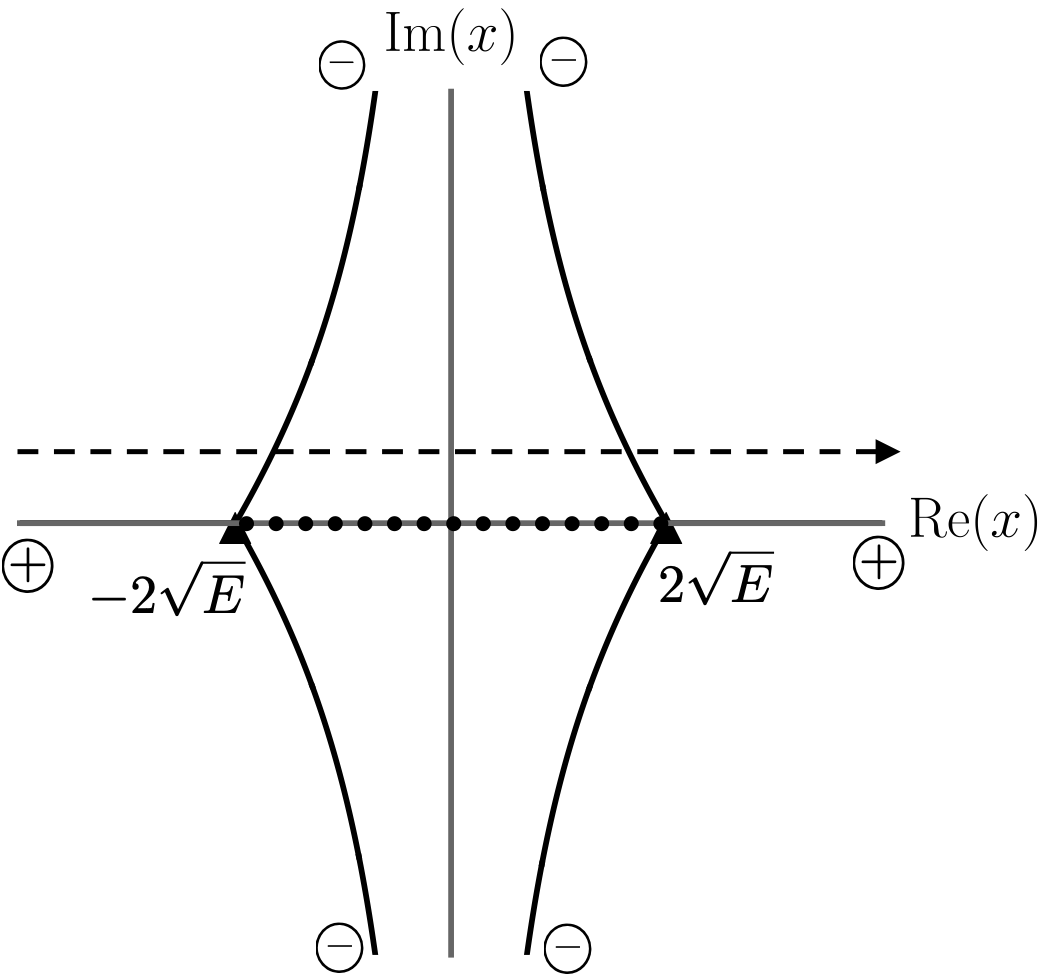}
    \caption{The Stokes geometry of the potential $Q(x)=x^2/4-E$ on the complex $x$ plane, with horizontal and vertical axes being those of ${\rm Re} \, (x)$ and ${\rm Im} \, (x)$, respectively.
    The triangles denote the turning points.
    To determine the branch, we choose $\sqrt{Q(x)}>0$ for $x>2\sqrt{E}$. The branch cut is shown as the dotted line. The dominant WKB solution, $\psi^{\rm WKB}_{+}$ or $\psi^{\rm WKB}_{-}$, on each Stokes curve is denoted by \circled{+} or \circled{-},  respectively. We use these symbols throughout this paper. The dashed arrow line describes the path of analytic continuation.}
    \label{fig:SL_harmonic}
\end{figure}

Next, let us consider the analytic continuation from $x=-\infty$ to $x = +\infty$. From figure \ref{fig:SL_harmonic}, we observe that these asymptotic points are on the Stokes curves. To avoid the failure of the analysis due to the Borel non-summability on the Stokes curves, let us deform the contour to being slightly above the real axis, i.e.~from $x=-\infty +i \epsilon$ to $x=+\infty + i \epsilon$ with $\epsilon >0$.%
\footnote{The deformation of the contour slightly below the real axis gives a slightly different relation between the wave functions but results in the same energy eigenvalues.}
We take $\epsilon \to 0$ at the end of calculations.
The contour must then cross two Stokes curves, the one emanating from $x=x_1$ and then the other from $x= x_2$, as can be seen in figure \ref{fig:SL_harmonic}. Thus, one needs to consider the Stokes phenomena twice and the effect of changing the reference points. According to the general rules reviewed in section \ref{subsec:StokesLines_Connection}, the analytic continuation of the WKB solutions is given by
\begin{align}\label{eq:AC_harmonic}
    \begin{pmatrix}
        \psi_+\\
        \psi_-
    \end{pmatrix}^{(-\infty)}
    =M
    \begin{pmatrix}
        \psi_+\\
        \psi_-
    \end{pmatrix}^{(+\infty)}~,
\end{align}
where the superscripts $\pm\infty$ indicate that $\psi_\pm$ are constructed in the Stokes regions that contain $x = \pm\infty$, respectively, and 
\begin{align}
    M=(V_-)^{-1}\Gamma_{12}(V_-)^{-1}\Gamma_{21}
    =
    \begin{pmatrix}
        1 & \quad0\\
        -i(1+u_{21}^2) &\quad 1
    \end{pmatrix}~.
\end{align}
In this computation, the inverse of $V_-$ is taken since the Stokes curves are crossed clockwise around the corresponding turning points, while $\Gamma_{12}$ and $\Gamma_{21}$ represent the effects of changing the reference points from the turning point $x_1$ to $x_2$, and then from $x_2$ to $x_1$, respectively. The $\Gamma_{21}$ operation is not necessary to obtain the condition for eigenvalues below.

To obtain the energy spectrum, we impose boundary conditions such that the solution is exponentially suppressed at $x=\pm\infty$. 
Since $\psi_-$ is subdominant to $\psi_+$  in either end, we require $\psi_+$ to vanish at both limits.
From the formula~\eqref{eq:AC_harmonic}, the solution which satisfies boundary condition at $x=-\infty$, which is the solution purely made from $\psi_{-}^{(-\infty)}$, becomes
\begin{align}
    M_{21}\psi_+^{(+\infty)} + M_{22} \psi_-^{(+\infty)}~,
\end{align}
around $x=+\infty$, up to an overall integration constant. Thus, in order for the solution to satisfy the boundary conditions at $x=\pm\infty$, i.e.~for the coefficient of $\psi_+^{(+\infty)}$ to vanish, we need to impose $M_{21}=0$. This requirement corresponds to the condition $u_{21}^2 = -1$.
We can find $u_{21}^2$ by 
\begin{align}
    \label{eq:u21_harmonic}
    u_{21}^2&=\exp\bigg(-\oint_{\gamma_{12}}S_{\text{odd}}\, \dd x \bigg)
    =\exp\bigg(2\pi i \Res{x}{\infty}S_{\text{odd}}\bigg)\notag\\
    &=\exp\bigg(2\pi i\eta \Res{x}{\infty}\sqrt{Q(x)}\bigg)
    =\exp\bigg(2\pi i \eta E\bigg)~,
\end{align}
where the integration contour $\gamma_{12}$ is defined in figure \ref{fig:gamma_ij}.
In the second equality, we use that the residue theorem for a closed contour can equivalently be expressed in terms of the residues outside the contour, with the orientation properly accounted for. The residue at the infinity, which is a singular point, can be computed by changing variables from $x$ to $y=1/x$, giving $S_{\rm odd}(x) \, \dd x = - y^{-2} S_{\rm odd}(1/y) \, \dd y$.
In the third equality, we use the property that $S_{2j-1} \, (j\geq 1)$ are holomorphic at $x=\infty$,%
\footnote{This is the characteristic of irregular singular points.}
and $S_{-1} = \eta \sqrt{Q(x)}$ is the only term that has a residue there, 
which can be easily checked from \eqref{eq:S_j}.
Therefore, the exact WKB analysis shows that, requiring $u_{21}^2=-1$, energy eigenvalues must be
\begin{align}
    E =\eta^{-1}\left(n+ \frac{1}{2}\right)~, \qquad n=0,1,2,\cdots~,
\end{align}
which nicely agrees with the known energy spectrum of the harmonic oscillator~\eqref{eq:eigen_harmonic}.
We note that one only needs to know the analytic structure of the given problem in this method, without any need of explicitly solving the differential equation.

\subsection{Morse potential}
\label{subsec:Morse_potential}

\begin{figure}[t]
    \centering
    \includegraphics[width=0.5\linewidth]{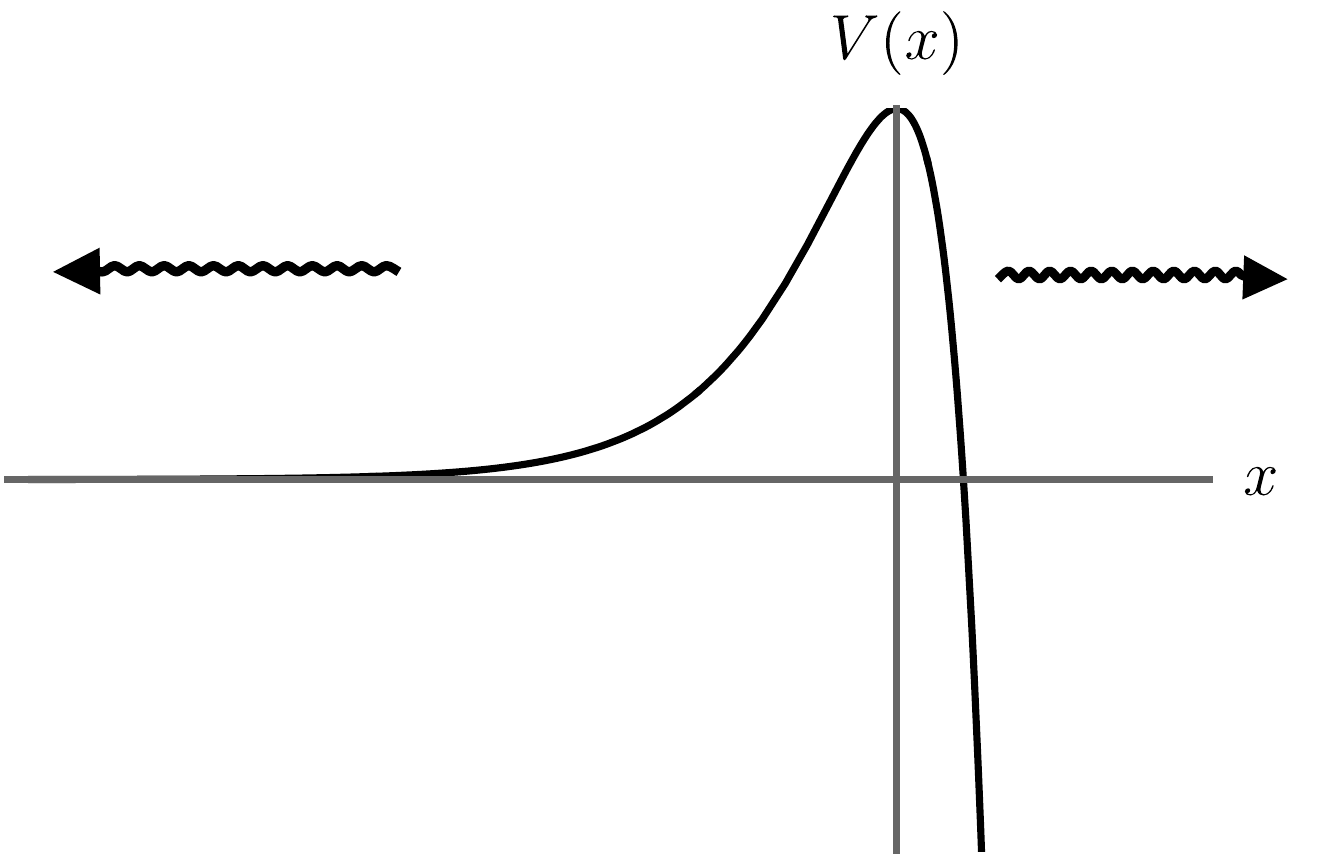}
    \caption{The schematic picture of the Morse potential in $x$ coordinate. We take the left-moving boundary condition at $x=-\infty$ (leftwards arrow wave) and the right-moving boundary condition at $x=\infty$ (rightwards arrow wave), defined below eq.~\eqref{eq:sol_morse_asymptote}.}
    \label{fig:potential_Morse}
\end{figure}
Our next example is the Morse potential, which is closer to the case of black hole perturbations than the previous one
(see~\cite{Hatsuda:2021gtn} for a review).
The corresponding differential equation is given by
\begin{align}
    \label{eq:eom_Morse_original}
    \left[ -\frac{\dd^2}{\dd x^2}+V(x) \right] \Psi=E\Psi~,
\end{align}
where the potential $V$ is
\begin{align}
    V(x)&=V_0 \left( 2 \, {\rm e}^{\beta x} - {\rm e}^{2\beta x} \right)~,
\end{align}
with $V_0>0$ and $\beta>0$ constant parameters and $-\infty<x<\infty$. See figure \ref{fig:potential_Morse} for a schematic figure of $V(x)$.
As can be seen from the figure, the Morse potential appears quite different in its asymptotic behavior from the the potential for perturbations in the Schwarzschild spacetime, which is our target in this work and is considered in section \ref{sec:Schwarzschild}. Nevertheless, in the context of the exact WKB analysis, it is appropriate as a simplified example for black hole QNMs for two reasons: (i) eq.~\eqref{eq:eom_Morse_original} is exactly solvable as seen below, and (ii) the structure of singular points of the differential equation is similar to that of the spectral problems of asymptotically flat black holes \cite{Hatsuda:2021gtn}. 
In particular, the Stokes geometry of the Morse potential contains one regular and one irregular singular point, and due to the presence of the former, a logarithmic spiral into the singularity appears, which makes it suitable as a warm-up for the black hole case.
The exact WKB technique is to reduce a problem of differential equations to that of their analytic structure, and thus the studies of systems that share analytically similar behaviors proceed in a parallel manner.

To carry on our analysis, we find it convenient to rescale the parameters as
\begin{align}
    \epsilon := \frac{E}{V_0} \; , \qquad
    g := \frac{\beta}{\sqrt{V_0}} \; ,
\end{align}
and to change variables by
\begin{align}
    z := 2 \, {\rm e}^{\beta x} \; ,
\end{align}
where the original domain $-\infty < x < \infty$ maps to $0<z<\infty$.%
\footnote{These rescaling and redefinition are chosen to make convenient the exact WKB analysis we conduct below, see eq.~\eqref{eq:eom_Morse_canonical}.}
The dimensionless parameter $g$ is real and positive, while the rescaled eigenvalue $\epsilon$ can take a complex value.
The resulting differential equation reads
\begin{align}
    \label{eq:morseschrodingereq}
    \left[ - z \, \frac{\dd}{\dd z} \left( z \, \frac{\dd}{\dd z} \right) + \frac{1}{g^2} \left( - \frac{z^2}{4} + z \right) \right] \Psi = \frac{\epsilon}{g^2} \, \Psi \; .
\end{align}
This equation has one regular singularity at $z=0$ (i.e.~$x=-\infty$) and an irregular singularity at $z=\infty$ (i.e.~$x=\infty$). We later see the presence of the regular singularity cause some complications in the exact WKB analysis, which also occurs in the case of the perturbation around the black hole background.

Before applying the exact WKB analysis, let us make use of the fact that the Morse problem is exactly solvable and
solve it using the conventional method by special functions. 
Equation~\eqref{eq:morseschrodingereq} can be solved by 
the confluent Hypergeometric functions (a.k.a.~the Kummer functions, see~\cite{NIST:DLMF} for their definitions and properties). The two independent solutions are
\begin{align}
    \Psi_1 & = \ee^{\frac{i y}{2}} y^{- \frac{i \sqrt{\epsilon}}{g}} \, 
    M(a,b, -i y) \; , \\
    \Psi_2 & = \ee^{\frac{i y}{2}} y^{- \frac{i \sqrt{\epsilon}}{g}} \, 
    U(a,b, -i y) \; ,
\end{align}
where
\begin{align}
    a = \frac{1}{2} + i \, \frac{1-\sqrt{\epsilon}}{g} \; , \qquad
    b = 1- \frac{2i \sqrt{\epsilon}}{g} \; , \qquad
    y = \frac{z}{g} \; .
\end{align}
In order to demand boundary conditions at $z = 0$ and $z=\infty$ (i.e.~boundaries infinitely far away in terms of the original variable $x$), we expand $\Psi_{1,2}$ around these asymptotic limits, giving
\begin{subequations}
\label{eq:sol_morse_asymptote}
\begin{align}
    \label{eq:sol_morse_asymptote_1}
    \Psi_1 &\to 
    \begin{cases}
        \displaystyle
       y^{- \frac{i \sqrt{\epsilon}}{g}}~, \qquad & (z \to 0)~, \vspace{2mm} \\
       \displaystyle
       \left( - i \right)^{a - b} \, \frac{\Gamma (b)}{\Gamma (a)} \, {\rm e}^{- \frac{iy}{2}} y^{- \frac{1}{2} + \frac{i}{g}} + i^{-a} \, \frac{\Gamma(b)}{\Gamma(b-a)} \, {\rm e}^{\frac{iy}{2}} y^{- \frac{1}{2} - \frac{i}{g}}~, \qquad & (z \to \infty)~,
    \end{cases} \\
    \label{eq:sol_morse_asymptote_2}
        \Psi_2 &\to 
    \begin{cases}
    \displaystyle
      \frac{\Gamma (1 - b)}{\Gamma ( a - b +1 )} \, y^{- \frac{i \sqrt{\epsilon}}{g}} 
      + \left( - i \right)^{1-b} \, \frac{\Gamma (b-1)}{\Gamma (a)} \, y^{\frac{i \sqrt{\epsilon}}{g}} ~, \qquad & (z \to 0)~, \vspace{2mm}\\
      \displaystyle
      \left( - i \right)^{-a} \, {\rm e}^{\frac{iy}{2}} \, y^{- \frac{1}{2} - \frac{i}{g}} ~, \qquad & (z \to \infty)~.
    \end{cases}
\end{align}
\end{subequations}
As boundary conditions, we impose the solution behaves as $\sim y^{- i \sqrt{\epsilon} / g}$ (we call it ``left-moving'') at the boundary $z = 0$ and as $\sim {\rm e}^{i y / 2}$ in the limit $z \to \infty$ (we call it ``right-moving'').%
\footnote{Note that the terms ``left-'' and ``right-moving'' used here are merely nomenclature and do not necessarily carry physical meanings.}
Seen from eq.~\eqref{eq:sol_morse_asymptote}, $\Psi_1$ is purely ``left-moving'' near $z=0$ (i.e.~$x=-\infty$), while $\Psi_2$ is purely ``right-moving'' around $z=\infty$ (i.e.~$x=+\infty$).
In this case, therefore, the solution satisfies the boundary conditions at both ends if and only if $\Psi_1$ and $\Psi_2$ are not linearly independent of each other. This is possible when the Wronskian between $\Psi_1$ and $\Psi_2$ vanishes.
From the property of the confluent Hypergeometric functions, the Wronskian is given by
\begin{align}
    W_{12} 
    & = \frac{\dd \Psi_1}{\dd z} \, \Psi_2 - \frac{\dd \Psi_2}{\dd z} \, \Psi_1
    = \ee^{iy} y^{- \frac{2i \sqrt{\epsilon}}{g}}
    \left(
    \frac{\dd M}{\dd z} \, U - \frac{\dd U}{\dd z} \, M
    \right)
    = \frac{-i}{z} \, \frac{\Gamma(b)}{\Gamma(a)} \; ,
\end{align}
which becomes zero only when the argument of the Gamma function in the denominator takes non-positive integers, that is, $a = \frac{1}{2} + i \, \frac{1 - \sqrt{\epsilon}}{g} = -n$ for $n=0,1,2,\dots$. In other words, the eigenvalues of $\epsilon$ are given by
\begin{align}
    \epsilon_n =\bigg[1-ig\bigg(n+\frac{1}{2}\bigg)\bigg]^2~, \qquad 
    n =0, \, 1, \, 2,\ldots~,
    \label{eq:Eigenvalue_Morse}
\end{align}
where $\sqrt{\epsilon_n}$ should be evaluated such that $1-\sqrt{\epsilon_n}$ is purely imaginary and ${\rm Im}\big( 1-\sqrt{\epsilon_n} \big) >0$.

In the case of ${\rm Re}\big(\sqrt{\epsilon}\big) < 0$, $\Psi_1$ is no longer eligible to meet the boundary condition at $z=0$. In order to impose the outgoing boundary conditions on $\Psi_2$ both at $z=0$ and at $z=\infty$, we now set the coefficient of the $y^{-i \sqrt{\epsilon}/z}$ term to $0$ in eq.~\eqref{eq:sol_morse_asymptote_2}. This occurs when $a - b + 1 = \frac{1}{2} + i \, \frac{1 + \sqrt{\epsilon}}{g}$ is a non-positive integer, resulting in
\begin{align}
    \label{eq:Eigenvalue_Morse_negative}
    \epsilon_n = \left[ - 1 + i g \left( n + \frac{1}{2} \right) \right]^2 \; , \qquad
    n = 0, \, 1, \, 2, \ldots~,
\end{align}
where $\sqrt{\epsilon_n}$ should be taken to be negative.
Although eqs.~\eqref{eq:Eigenvalue_Morse} and \eqref{eq:Eigenvalue_Morse_negative} look identical, they differ by the branch of $\sqrt{\epsilon_n}$ and by their corresponding wave functions.

Now, let us apply the exact WKB method and reproduce the above eigenvalues.
To bring eq.~\eqref{eq:morseschrodingereq} to the Schr\"{o}dinger type, we change variables from $\Psi$ to  $\psi=z^{1/2}\Psi$. Then eq.~\eqref{eq:morseschrodingereq} recasts into
%
\begin{align}
    \label{eq:eom_Morse_canonical}
    \left[ - \frac{\dd^2}{\dd z^2} + \frac{1}{g^2} \left( - \frac{1}{4} + \frac{1}{z} - \frac{\epsilon}{z^2} \right) - \frac{1}{4z^2} \right] \psi = 0 \; .
\end{align}
To construct the formal WKB series solutions, we introduce the expansion parameter $\eta$ as follows:
\begin{align}
	\left[-\frac{\dd^{2}}{\dd z^{2}} + \eta^{2}Q(z,\eta) \right]\psi = 0~,
    \label{eq:Schrodinger_type_Morse}
\end{align}
where
\begin{subequations}
\begin{align}
	&Q(z, \eta) = Q_{0}(z) + \eta^{-2}Q_{2}(z)~, \label{eq:Q_expand_Morse}\\
    \label{eq:Morse_Q0}
	&Q_{0}(z) = \frac{1}{g^2} \left( - \frac{1}{4} + \frac{1}{z} - \frac{\epsilon}{z^2} \right)~,\\
	&Q_{2}(z) =  -\frac{1}{4z^{2}} ~.
\end{align}
\end{subequations}
We insert $\eta$ such that the asymptotic behaviors of the leading WKB solutions to eq.~\eqref{eq:Schrodinger_type_Morse} around the singularities (with $\eta$ taken) match with those obtained from eq.~\eqref{eq:eom_Morse_canonical}. 
The former is found by expanding $\psi^{\rm WKB}_\pm \sim Q_0^{-1/4} \exp \left( \pm \int^z \sqrt{Q_0} \, \dd z \right)$ around the singular points $z = 0$ and $\infty$, while the latter is obtained by inserting $Q(z,1) \simeq - \left( \frac{\epsilon}{g^2} + \frac{1}{4} \right) \frac{1}{z^{2}}$ around $z = 0$ and $Q(z,1) \simeq \frac{1}{g^2} \left( - \frac{1}{4} + \frac{1}{z} \right)$ around $z=\infty$ into eq.~\eqref{eq:Schrodinger_type_Morse} with $\eta=1$, and then by solving it in those limits. One should include the subleading-order term for $z=\infty$ in order to correctly capture the non-suppressed contributions to the asymptotic solutions.
With the choice of $Q_0$ in eq.~\eqref{eq:Morse_Q0}, the leading WKB solutions behave as 
\begin{align}
    \label{eq:sol_WKB_morse_asymptotic}
    \psi^{\text{WKB}}_{\pm} \sim
    \begin{cases}
        \displaystyle
        z^{\frac{1}{2}\pm \frac{i \eta \sqrt{\epsilon}}{g}}~,\qquad & z\to 0~, \vspace{2mm}\\
        \displaystyle
        \ee^{\pm \frac{i \eta}{2 g} z} z^{\mp \frac{i \eta}{g}}~, \qquad & z\to \infty~,
    \end{cases}
\end{align}
up to constant prefactors.
Here, we draw the branch cut originated from $\sqrt{Q_0}$ as is between the two turning points, shown in figure \ref{fig:SL_schematic_Morse}, and take the branch of $\sqrt{Q_0}$ in the asymptotics to be
\begin{align}
    \sqrt{Q_0}\simeq
    \begin{cases}
        \displaystyle
        \frac{i \sqrt{\epsilon}}{gz}~, \qquad & ( z\to 0)~, \vspace{2mm} \\
        \displaystyle
        \frac{i}{g} \left( \frac{1}{2}-\frac{1}{z} \right)~, \qquad & (z\to \infty)~.
    \end{cases}
    \label{eq:branch_Morse}
\end{align}
These asymptotic WKB approximations \eqref{eq:sol_WKB_morse_asymptotic} correctly capture the behaviors of the solutions to the differential equation \eqref{eq:Schrodinger_type_Morse}, obtained by solving it with $Q$ replaced by its asymptotic forms.
By our parametrization, $\eta$ happens to always appear with $1/g$, and one may as well use $g$ as the expansion parameter. Let us, however, emphasize that this is merely a coincidence by our choice, and in general, the placement of $\eta$ needs to be done with care by analyzing appropriate asymptotic behaviors.
Given \eqref{eq:sol_WKB_morse_asymptotic}, the desired boundary conditions, described below eq.~\eqref{eq:sol_morse_asymptote}, require that the solution behave as $\psi_{-}^{\text{WKB}}$ at $z=0$ and $\psi_{+}^{\text{WKB}}$ at $z=\infty$. 

The dominance relations on a Stokes curve that flows into a regular singular point can be determined from the residue of $\sqrt{Q_0(z)}$ at the same singular point.
For, near a regular singular point $z=s_r$, we have
\begin{align}
\label{eq:dominance_Morse}
    \int^z \sqrt{Q_0(z)} \, \dd z \simeq
    \int^z \frac{\Res{z}{s_r} \sqrt{Q_0(z)}}{z-s_r} \, \dd z
    \simeq \Res{z}{s_r} \sqrt{Q_0(z)} \, \log z \; ,
\end{align}
and the dominance relation is defined by the real part of this quantity as in Eqs.~\eqref{eq:dominance_positive}, and \eqref{eq:dominance_negative}.
Since the residue at $z=0$ is $\Res{z}{0}\sqrt{Q_0(z)}=i\sqrt{\epsilon}/g$ with its real part positive and $\log \vert z \vert < 0$, $\psi_-$ is dominant on the Stokes curves flowing into $z=0$.%
\footnote{This is because we take $\sqrt{\epsilon} = \sqrt{\epsilon_0} = 1 - ig /2$ in figure \ref{fig:SL_schematic_Morse}, and thus ${\rm Im}\big( \sqrt{\epsilon} \big) < 0$. In the opposite case of ${\rm Im}\big( \sqrt{\epsilon} \big) > 0$, the dominance relations flip, and one has to conduct the subsequent calculations accordingly.}
Once the dominance relation on these curves is determined, the dominant solution on each of the Stokes curves emanating from $z=z_2$ is fixed automatically, as the dominance relations alternate among the Stokes curves emanating from a single turning point. This is consistent with our choice of branch \eqref{eq:branch_Morse} and of the branch cut in figure \ref{fig:SL_schematic_Morse} which shows the Stokes curves for the Morse potential drawn from $Q_0$ with the eigenvalue $\epsilon=\epsilon_0$. 

\begin{figure}
    \begin{tabular}{cc}   
        \begin{minipage}[b]{0.45\linewidth}
            \centering
            \includegraphics[keepaspectratio, scale=0.3]{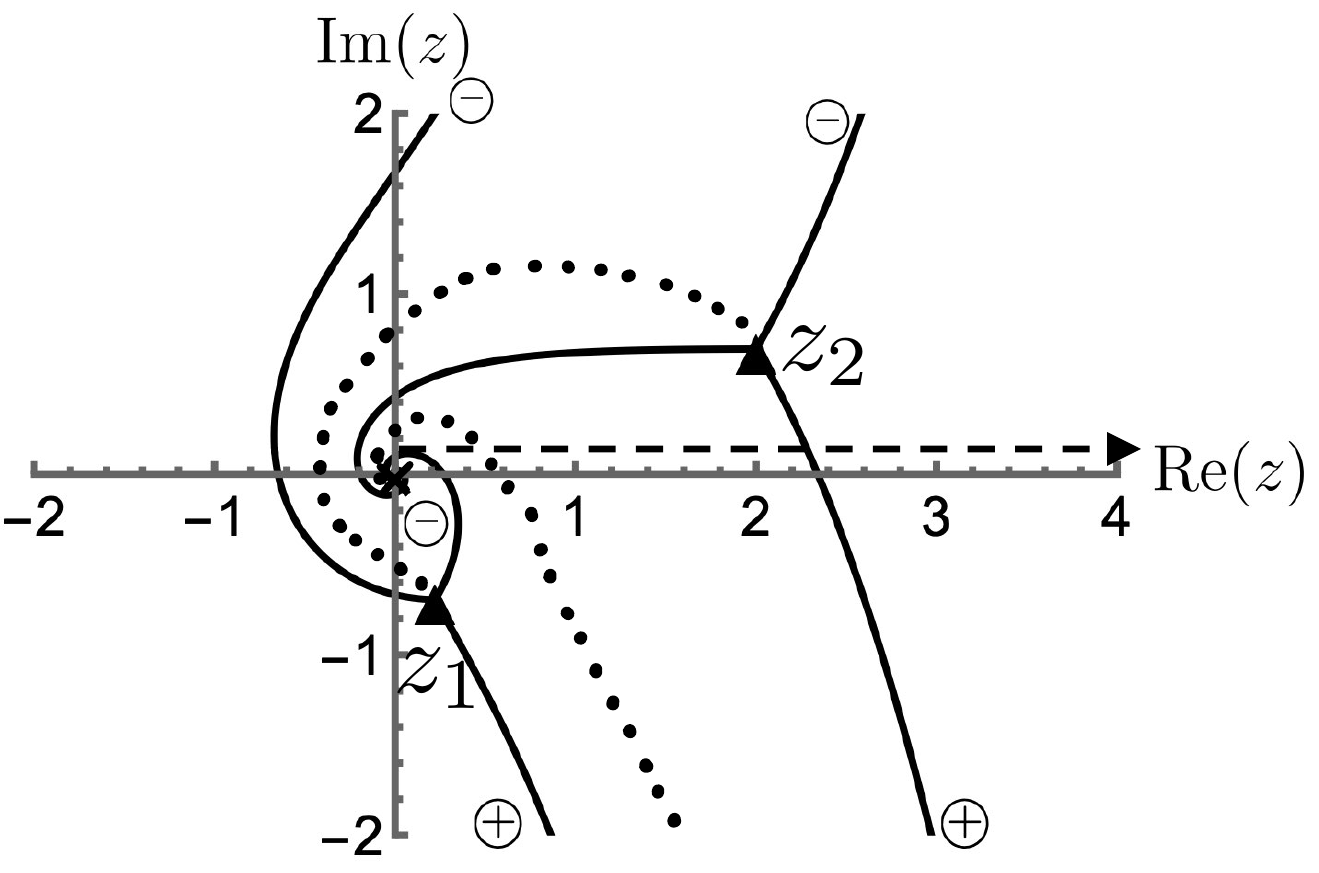}
        \end{minipage}&
        \begin{minipage}[b]{0.45\linewidth}
            \centering
            \includegraphics[keepaspectratio, scale=0.33]{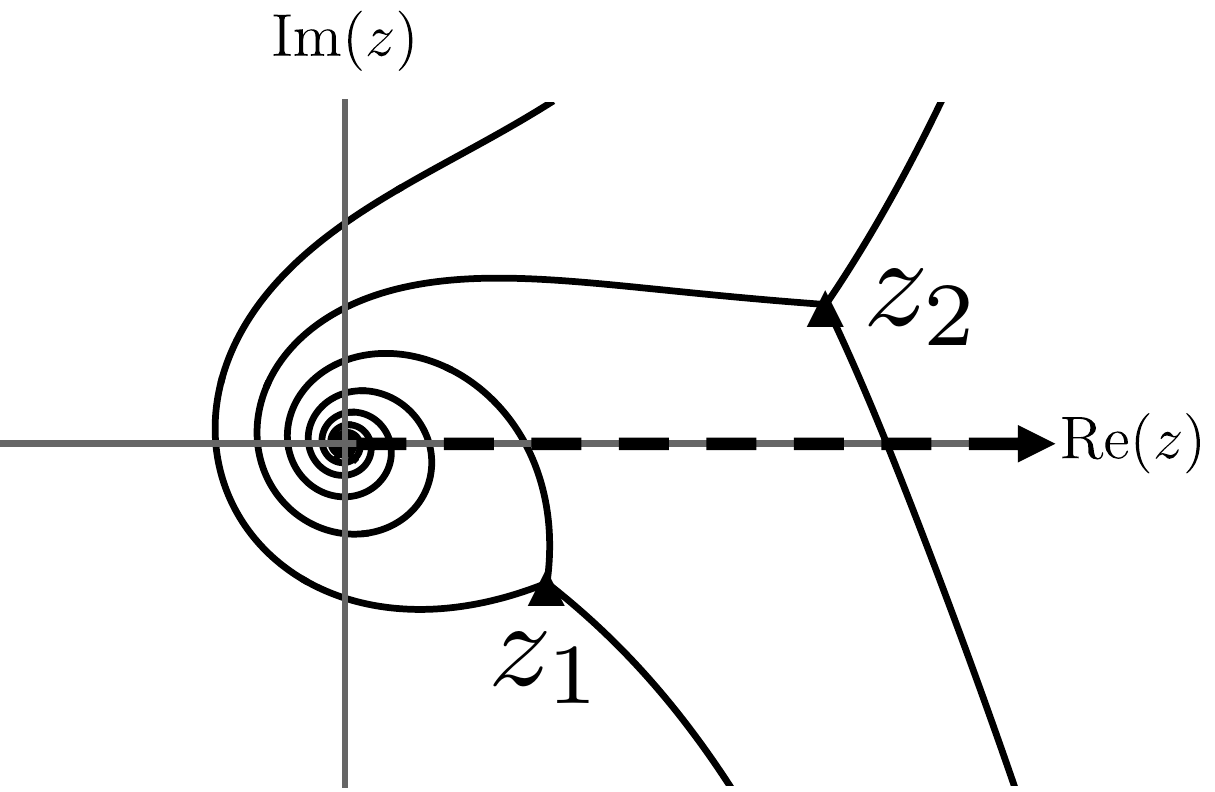}
        \end{minipage}
    \end{tabular}
    \caption{ (Left) The Stokes curves for the potential \eqref{eq:Morse_Q0} with $\epsilon=\epsilon_0=(1-ig/2)^2$ and $g=1$. The triangular points $z_1$ and $z_2$ represent the turning points. The dotted curves denote our choice of branch cuts from a singular point and between turning points. We draw the symbol \textcircled{+} (\textcircled{-}) at the end of each Stokes curve if $\psi_+^{\rm WKB}$ ( $\psi_-^{\rm WKB}$) is dominant on it. The dashed line denotes the path of analytic continuation on the real positive axis.
    (Right)  The schematic picture of the Stokes curves for the Morse potential. We observe the logarithmic spiral flow of the Stokes curves near the regular singularity at $z=0$.
    For a clear illustration, we suppress showing the branch cut between $z=0$ and $z=\infty$, which runs along the logarithmic spiral.}
    \label{fig:SL_schematic_Morse}
\end{figure}
An extra care is needed as we see that there is a logarithmic spiral flow of the Stokes curves near the regular singularity at $z=0$ (see the right panel of figure \ref{fig:SL_schematic_Morse}).
Using \eqref{eq:SL_definition}, \eqref{eq:branch_Morse} and \eqref{eq:dominance_Morse}, and introducing the polar coordinate $z = r\ee^{i\theta}$, we obtain
\begin{align}
    \log r \simeq \frac{\text{Im}\sqrt{\epsilon}}{\text{Re}\sqrt{\epsilon}} \, \theta~.
\end{align}
When approaching $r \to 0$, $\vert \theta \vert $ takes an arbitrarily large value, exhibiting the presence of a logarithmic spiral.
We emphasize that this is a feature of regular singular points and does not appear at irregular singular points.

Let us consider the analytic continuation of the WKB solution $\psi_-^{\rm WKB}$ at $z=0$ into $z=\infty$ along the positive real axis of $z$.
Due to the logarithmic spiral flow of the Stokes curve, the contour encounters an infinite number of intersections with the Stokes curves. That is, we have
\begin{align}
    \begin{pmatrix}
        \psi_+\\
        \psi_-
    \end{pmatrix}^{(0)}
    =
    M
    \begin{pmatrix}
        \psi_+\\
        \psi_-
    \end{pmatrix}^{(\infty)}~,
\end{align}
where
\begin{align}
    M
    &=\lim_{n\to\infty}\left[ \left( V_- \right)^{-1} \Gamma_{21}\left( V_- \right)^{-1} \Gamma_{12} D_0 \right]^n V_+ \cr
    &=\lim_{n\to\infty}
    \begin{pmatrix}
        \displaystyle
        \left( \nu^+_0 \right)^n & \quad
        \displaystyle 
        i \left( \nu^+_0 \right)^n \vspace{2mm} \\
        \displaystyle
        -i\nu_0^+ \, \frac{\left( \nu_0^- \right)^n - \left( \nu_0^+ \right)^n}{\nu_0^- -\nu_0^+} \left( 1+u_{12}^2 \right) & \quad
        \displaystyle
        \left( \nu^-_0 \right)^n + \nu_0^+ \, \frac{\left(\nu_0^- \right)^n - \left( \nu_0^+ \right)^n}{\nu_0^- -\nu_0^+} \left( 1+u_{12}^2 \right)
    \end{pmatrix}~.
\end{align}
Here, $V_-$ is the effect of crossing the Stokes curve on which $\psi_-$ is dominant as in eq.~\eqref{eq:connectionS}, $\Gamma_{12} = \Gamma_{21}^{-1}$ that of changing the reference points of the integral from the turning point $z_1$ to the other $z_2$ as in eq.~\eqref{eq:connectionT}, and $D_0$ that of crossing the branch cut attached the regular singular point at $x=0$ as in \eqref{eq:MatrixD}, while $\nu^{\pm}_0$ and $u_{12}$ are given by, according to eqs.~\eqref{eq:crossing_branchcut} and \eqref{eq:uij},
\begin{align}
    \nu^{\pm}_0 &=\exp\left[ i\pi \left( 1 \pm 2 \, \Res{z}{0} S_{\text{odd}} \right) \right]~,\\
    u_{12}^2 &= \exp\left( \oint_{\gamma_{12}} S_{\text{odd}} \, \dd z \right) = \exp\left( -2\pi i\Res{z}{0,\infty}S_{\text{odd}} \right) 
    \label{eq:u12Morse}~.
\end{align}
The integration contour $\gamma_{12}$ is defined in figure \ref{fig:gamma_ij}.
In the second equality of eq.~\eqref{eq:u12Morse}, the closed contour $\gamma_{12}$ encircling the branch cut between the two turning points counterclockwise is identified with the one encircling the outside clockwise, containing the singular points at $z=0$ and $+\infty$. 
Since we have (see appendix \ref{app:regular_residue} for the derivation)%
\footnote{From the explicit expressions of each term in $S_{\rm odd}$, one can see that $S_{-1}$ is the only term that contributes to the residue both at $z=0$ and $z=\infty$.}
\begin{align}
    \Res{z}{0}\,S_{\text{odd}} = \frac{i \eta \sqrt{\epsilon}}{g}~, \qquad
    \Res{z}{\infty}\,S_{\text{odd}}=\frac{i \eta}{g}~,
\end{align}
we obtain
\begin{align}
    &\nu^{\pm}_0 =\exp \left(i\pi\mp2\pi \, \frac{\eta \sqrt{\epsilon}}{g} \right)~,\\
    &u_{12}^2
    =\exp\left[ \frac{2\pi \eta}{g} \left( \sqrt{\epsilon}+1 \right) \right]~.
\end{align}

The solution satisfies the ``left-'' and ``right-moving'' boundary conditions, which are to take $\psi_-$ at $z=0$ and $\psi_+$ at $z = \infty$, if $M_{22}=0$, i.e.,
\begin{align}\label{eq:MorseQNMcondition}
    \lim_{n\to\infty} \left[ \left( \nu_0^- \right)^n \left( 1+\nu_0^+ \, \frac{1-\left( \frac{\nu_0^+}{\nu_0^-} \right)^n}{\nu_0^- - \nu_0^+}\left( 1+u_{12}^2 \right) \right) \right] = 0~.
\end{align}
The overall factor $\left( \nu_0^- \right)^n$ is degenerate with the constant coefficient of $\psi_-$ and can be absorbed into an integration constant.
Since $\text{Re}\big( \nu_0^+/\nu_0^- \big)<1$, we can neglect $(\nu_0^+/\nu_0^-)^n$ in the limit $n\to\infty$, and the condition eq.~\eqref{eq:MorseQNMcondition} is reduced to
\begin{align}
    1+\frac{\nu_0^+}{\nu_0^- - \nu_0^+}(1+u_{12}^2)=0~.
\end{align}
Rewriting this equation, we obtain
\begin{align}
    \exp\left[ \frac{2\pi \eta}{g} \left( \sqrt{\epsilon}-1 \right) \right] + 1 = 0~,
\end{align}
and finally, we get
\begin{align}
    \frac{2\pi \eta}{g} \left( \sqrt{\epsilon}-1 \right) = -2\pi i \left( n+\frac{1}{2} \right) \,,
\end{align}
where the sign of the right hand side is decided from the condition $\text{Im}\sqrt{\epsilon}<0$.
This equation gives the correct eigenvalues \eqref{eq:Eigenvalue_Morse} after setting $\eta = 1$.
In the case of ${\rm Re} \sqrt{\epsilon} < 0$, the same calculation reproduces \eqref{eq:Eigenvalue_Morse_negative} as well.

This example of the Morse potential shows a nontrivial demonstration of the usage of the exact WKB technique.
Moreover, we would like to remark on an important lesson that we must take into account the logarithmic spiral of Stokes curves when we consider the analytic continuation of the WKB solution from $z=0$ to $z=\infty$ on the positive real axis.
In general, logarithmic spirals appear around regular singularities of a given differential equation, and this phenomenon indeed occurs in the case of black hole spacetimes as we investigate in the next section.

\section{QNMs of Schwarzschild BH from exact WKB analysis}
\label{sec:Schwarzschild}
In this section, we apply the exact WKB analysis to QNMs problem in the Schwarzschild black hole spacetime, as our main goal of this paper.
The corresponding differential equation contains regular singular points, and hence the logarithmic spirals appear as in the case of the Morse potential.
We properly take them into account and impose the boundary conditions both at the event horizon and at the spatial infinity.
We demonstrate our formulation by computing the QNMs of perturbations around the Schwarzschild background metric fully analytically and, as a concrete result, reproduce the eigenvalues of higher overtones
\begin{equation}
\omega_n \to \frac{\log 3}{4 \pi} - \frac{i}{2} \left( n+\frac{1}{2} \right) \; ,
\end{equation}
in the limit $n \to \infty$. 
We take the normalization of $r_S = 2M = 1$ and $c = G_N = \hbar = 1$ throughout the section, where $M$ is the mass of the black hole.
The real part $\log 3 / (4\pi)$ is a highly nontrivial result, which was first conjectured in \cite{Hod:1998vk} based on \cite{Nollert:1993zz} and later confirmed by \cite{Motl:2002hd, Motl:2003cd,Andersson:2003fh}. We explicitly show that this result is obtained after imposing both of the boundary conditions and incorporating an infinite number of Stokes phenomena due to the logarithmic spirals. These mathematical objects indeed are of physical importance.

\subsection{Regge-Wheeler type equation}

We consider perturbing the Schwarzschild spacetime. Expanding perturbation variables in terms of the spherical harmonics and of a temporal component as $\sum {\rm e}^{-i \omega t} Y_{lm}(\theta,\phi) \, \Psi_l(r)$ up to appropriate powers of the radial coordinate $r$ (rescaled by $2 M$), the linearized radial equations for a massless scalar ($s=0$), electromagnetic field ($s=1$), and gravitational wave ($s=2$) are concisely written in a universal form of the Regge-Wheeler type
\begin{align}
	\bigg[f^2\frac{\dd^{2}}{\dd r^{2}} + ff' \frac{\dd}{\dd r} + \omega^{2} - V_{s}\bigg]\Psi_l = 0 \; ,
 \label{eq:Radial}
\end{align}
where prime denotes derivative with respect to $r$ and
\begin{align}
\label{eq:V_Schwarzschild}
    V_{s}=f\bigg[\frac{l(l+1)}{r^{2}} + \frac{1-s^2}{r^{3}}\bigg] \; , \qquad
    f(r)=1-\frac{1}{r} \; .
\end{align}
For gravitational waves, there exist two classes of modes, the Regge-Wheeler (odd/vector) \cite{Regge:1957td} and Zerilli (even/scalar) \cite{Zerilli:1970se,Zerilli:1970wzz}. However, it is known that the gravitational potentials of these two types are (super)partner potentials \cite{Cooper:1994eh}, and the corresponding radial solutions can be transformed into each other with the same QNM spectrum (isospectrality) \cite{Chandrasekhar1975,Chandrasekhar1984}. Except for so-called algebraically special modes, for which the transformation is singular, therefore, it suffices to analyze the Regge-Wheeler type equation \eqref{eq:Radial} in order to exhaust all the QNMs for $s=0,1,2$ (see \cite{Cardoso:2019mqo,Hatsuda:2021gtn} for reviews).
In fact, the analytic structure of eq.~\eqref{eq:Radial} changes for $s=0$ in the context of the exact WKB analysis, where the number of turning points reduces to three and the Schwarzschild singularity at $r=0$ becomes a simple pole rather than a higher-order singular point. As a result, the Stokes geometry drastically changes, and an independent analysis has to be conducted. 
We make a brief comment on this case at the end of this section.

To remove the first-order derivative term in eq.~\eqref{eq:Radial}, we change variables by $\psi = f^{1/2} \Psi_l$, which is the same as the procedure from eq.~\eqref{eq:general_secondorder_eq} to eq.~\eqref{eq:Schrodinger_type_eq}, and obtain a Schr\"{o}dinger-type equation as
\begin{align}
	\left( \frac{\dd^{2}}{\dd r^{2}} + \frac{\omega^{2}-V_{s}}{f^{2}} + \frac{f'{}^{2}}{4f^{2}} - \frac{f''}{2f} \right) \psi = 0 \; .
    \label{eq:mastereq_schwarzschild}
\end{align}
In order to construct the formal WKB series solution to eq.~\eqref{eq:mastereq_schwarzschild}, we introduce the WKB expansion parameter $\eta$ as follows
\begin{align}
	\left[ -\frac{\dd^{2}}{\dd r^{2}} + \eta^{2} Q(r,\eta) \right] \psi = 0 \; ,
    \label{eq:SchType}
\end{align}
where
\begin{subequations}
\label{eq:Q_expansion_schwarzschild}
\begin{align}
    Q(r,\eta) & = Q_{0}(r) + \eta^{-2}Q_{2}(r) \; ,
    \label{eq:Q_expand}\\
    Q_{0}(r) & = \frac{-\omega^{2}r^{4}+ l \left( l + 1 \right) r \left( r - 1 \right) - \left( r-1 \right) s^2}{r^{2}(r-1)^{2}} \; ,
    \label{eq:Q0}\\
    Q_{2}(r) & = \frac{-1}{4 \, r^{2} \left( r-1 \right)^{2}} \; ,
    \label{eq:Q_Schwarz}
\end{align}
\end{subequations}
and the formal series is constructed by expanding with respect to ``large'' $\eta$.
The original equation \eqref{eq:mastereq_schwarzschild} is recovered by taking $\eta = 1$.
This choice of $\eta$ insertion is justified from the viewpoint of the asymptotic behaviors, i.e., the asymptotic expansions of the leading-order WKB solutions match with those of the actual solutions to the differential equation \eqref{eq:mastereq_schwarzschild}.
The former is found by expanding $\psi^{\rm WKB}_\pm \sim Q_0^{-1/4} \exp \left( \pm \int^r \sqrt{Q_0} \, \dd r \right)$ around the singular points $r = 0,1,$ and $\infty$ (with $\eta=1$ taken), while the latter is by inserting 
\begin{align}
    \frac{\omega^{2}-V_{s}}{f^{2}} + \frac{f'{}^{2}}{4f^{2}} - \frac{f''}{2f} \simeq
    \begin{cases}
        \displaystyle
        - \frac{s^2 - 1/4}{r^2} \; , \qquad
        & r \to 0 \; , \vspace{2mm} \\
        \displaystyle
        \frac{\omega^2 + 1/4}{\left( r-1 \right)^2} \; , \qquad 
        & r \to 1 \; , \vspace{2mm} \\
        \displaystyle
        \left( 1 + \frac{2}{r} \right) \omega^2 \; , \qquad
        & r \to \infty \; ,
    \end{cases}
\end{align}
into eq.~\eqref{eq:mastereq_schwarzschild} and by solving it in the corresponding limits.%
\footnote{As in the case of the Morse potential in section~\ref{subsec:Morse_potential}, the subleading order should be included in the limit $r \to \infty$ due to the fact that it can make a diverging contribution to the solution in this limit, as can be seen in the last expression of eq.~\eqref{eq:psiWKB_limits}.}
\footnote{Since we impose boundary conditions only at the horizon ($r=1$) and at the spatial infinity ($r=\infty$), it is not strictly necessary for the two asymptotic solutions to also have the same behavior at $r=0$ for our purpose. Nonetheless, our choice of $Q_0$ and $Q_2$ in eq.~\eqref{eq:Q_expansion_schwarzschild} makes it possible for them to coincide around all the singular points $r=0,1,$ and $\infty$.}
The two general WKB solutions \eqref{eq:WKB_solutions} then read
\begin{align}
\label{eq:psiWKB_limits}
    \psi^{\text{WKB}}_{\pm} \sim
    \begin{cases}
        \displaystyle
        r^{1/2\pm s} \; , \qquad 
        & r\to 0 \; , \vspace{1mm}\\
        \displaystyle
        (r-1)^{1/2 \pm i\omega} \; , \qquad 
        & r\to 1 \; , \vspace{1mm}\\
        \displaystyle
        \ee^{\pm i\omega r} r^{\pm i\omega} \; , \qquad
        & r\to \infty \; ,
    \end{cases}
\end{align}
where we take the branch of $\sqrt{Q_0}$ as 
\begin{align}
    \sqrt{Q_0}\simeq
    \begin{cases}
        \displaystyle
        \frac{s}{r} \; , \qquad
        & r\to 0 \; , \vspace{2mm}\\
        \displaystyle
        \frac{i\omega}{r-1} \; , \qquad
        & r\to 1 \; , \vspace{2mm}\\
        \displaystyle
        \left( 1 + \frac{1}{r} \right) i\omega \; , \qquad
        & r\to \infty \; .
    \end{cases}
    \label{eq:branch_Schwarzchild}
\end{align}
To identify the QNM frequency $\omega$, we need to impose proper boundary conditions. In the community of black hole perturbations, it is customary to take them as incoming at the horizon $(r=1)$ and outgoing at infinity ($r=\infty$), which we follow in the following analysis.
In terms of WKB solutions, these conditions translate to taking $\psi_{-}^{\text{WKB}}$ in the Stokes region that contains $r=1$ and $\psi_{+}^{\text{WKB}}$ in the region with $r=\infty$, according to eqs.~\eqref{eq:psiWKB_limits}.

\subsection{Computing QNMs -- analytic continuation in the presence of logarithmic spirals}
\label{subsec:StokesLine_GR}
\begin{figure}[t]
    \centering
    \includegraphics[keepaspectratio, scale=0.7]{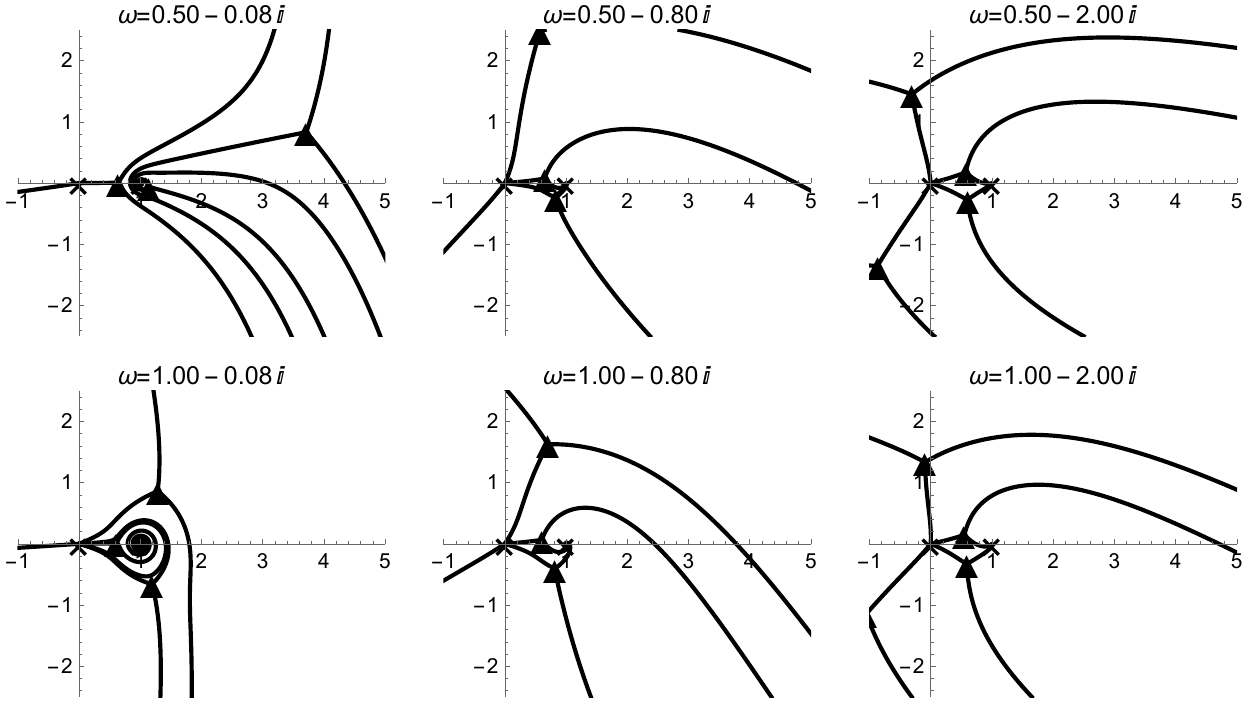}
    \caption{
    The Stokes curves for $s=l=2$ with various values of $\omega$. 
    Since the potential height is $\sim 0.6$ (in the units of $2 M$), we take two different values for ${\rm Re}(\omega)$, one smaller and one larger than 0.6, and for each of them, take three different values of ${\rm Im}(\omega)$, one small, one large, and one similar compared to the value of ${\rm Re}(\omega)$. The black crosses and solid triangles correspond to the singular points and turning points, respectively.
    Notice the presence of a logarithmic spiral around $r=1$ (and one around $r=0$). All of the $4$ turning points approach the origin as $\vert \omega \vert$ increases.
    }
    \label{fig:StokesLines_qnm}
\end{figure}

We now aim to find the QNMs of the perturbations around Schwarzschild black holes analytically using the exact WKB technique. To this end, we focus on higher overtones, which are characterized by negatively large ${\rm Im} \,(\omega)$, where analytical methods are most useful.
We promote the radial variable $r$ to a complex number, and essentially all the turning points, where we have $Q_0 = 0$, are located on the complex plane for the parameter values of our interest.
The Stokes curves for some representative values of QNM frequency $\omega$ are depicted in figure~\ref{fig:StokesLines_qnm}. These curves are drawn with respect to $Q_0$ under the criterion \eqref{eq:SL_definition}. We find that the topology of the {\it Stokes geometry} is insensitive to the value of $\omega$ as long as its imaginary part is large (compared to $s, \, l$ and $\vert {\rm Re} \, (\omega) \vert$) in the negative direction. The larger it is, the closer all the turning points approach to the origin $r=0$.

\begin{figure}[t]
    \begin{tabular}{cc}   
        \begin{minipage}[b]{0.45\linewidth}
            \centering
            \includegraphics[keepaspectratio, scale=0.3]{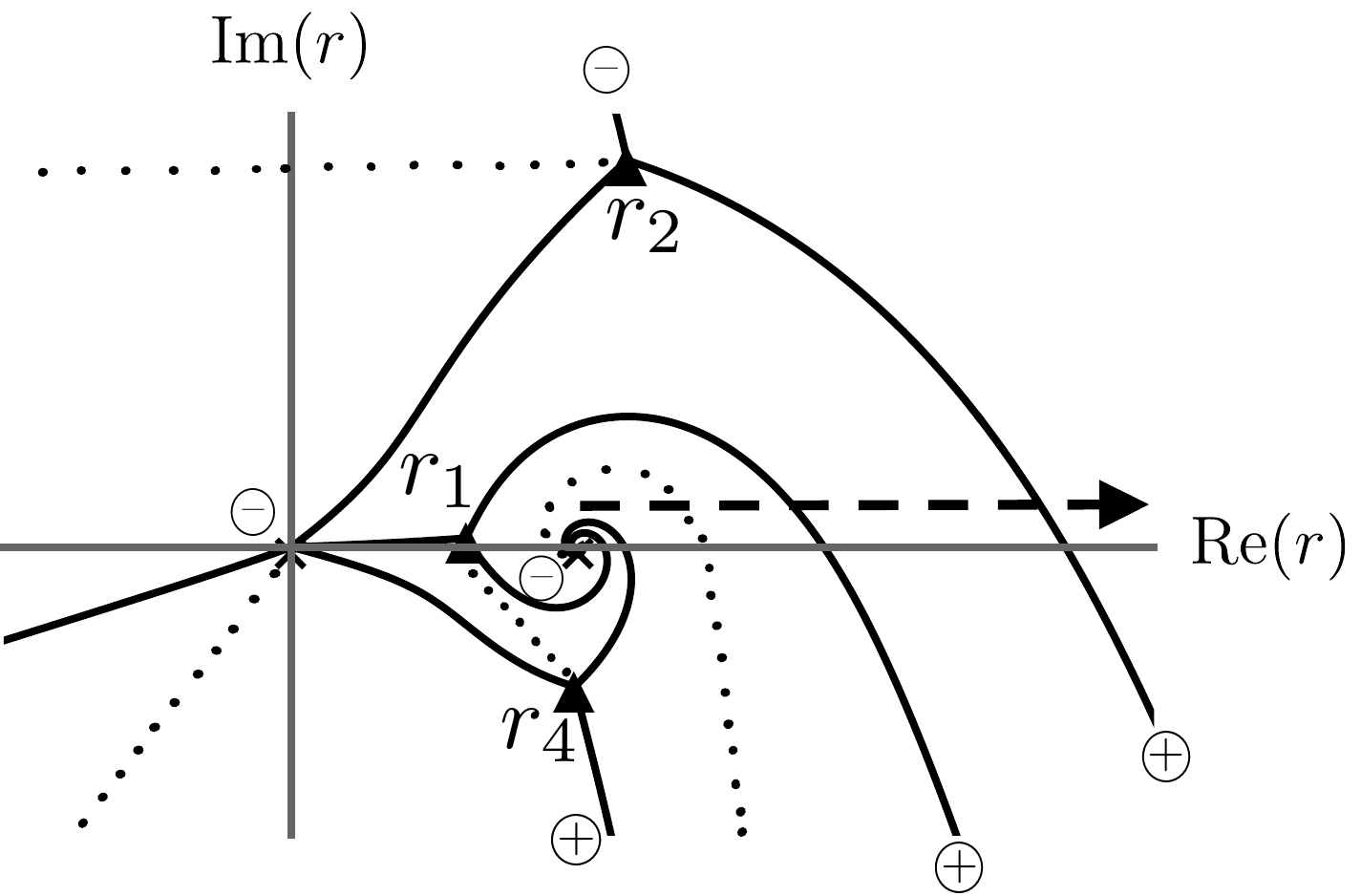}
        \end{minipage}&
        \begin{minipage}[b]{0.45\linewidth}
                \centering
                \includegraphics[keepaspectratio, scale=0.3]{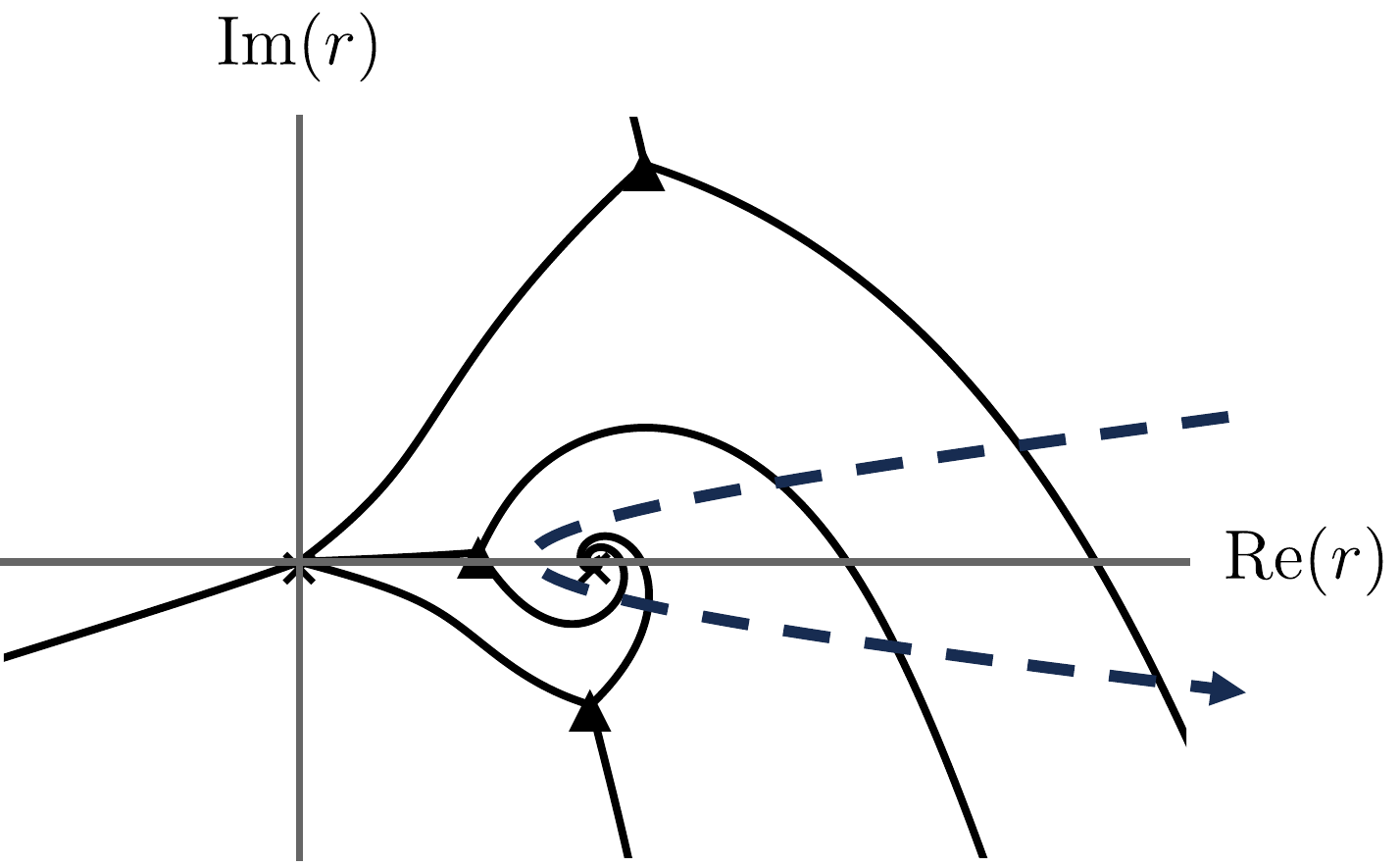}
        \end{minipage}
    \end{tabular}
    \caption{
    Schematic figures of Stokes curves for the potential \eqref{eq:Q0} with the branch cuts and the paths of analytic continuation. For illustration, we take $l=s=2$ and $\omega = 1 - 0.4 i$ (which is unrealistic but suitable to show the geometric features).
    (Left) The points $r_i\,(i=1,2,3,4)$ represent the turning points, and the dotted curves denote our choice of branch cuts between the turning points as well as those attached to the singular points. We place the symbol \textcircled{+} (\textcircled{-}) at the end of each Stokes curve on which $\Psi_+^{\rm WKB}$ ($\Psi_-^{\rm WKB}$) is dominant. The dashed line denotes the path of analytic continuation on the real positive axis.
    (Right) The same Stokes geometry with the path typically taken in the monodromy method \cite{Andersson:2003fh}.}
    \label{fig:SL_schematic_Schwarzschild}
\end{figure}

While almost invisible in the rightmost plots in figure \ref{fig:StokesLines_qnm}, there exists the logarithmic spiral flow of Stokes curves into the horizon ($r=1$) as well as into the curvature singularity ($r=0$), since these points are regular singular points of the differential equation \eqref{eq:SchType}. This feature is also present in the case of the Morse potential in section \ref{subsec:Morse_potential}. The exaggerated illustration of the logarithmic spiral into the horizon is shown in the left panel of figure~\ref{fig:SL_schematic_Schwarzschild}.

As discussed around eq.~\eqref{eq:dominance_Morse} for the Morse potential, we can infer the dominance relation of the Stokes curves from the residue of $\sqrt{Q_0(r)} \,$ at the regular singular points.
Since $\Res{r}{0}\sqrt{Q_0(r)}=s > 0$, $\psi_-$ is dominant on the Stokes curves flowing into $r=0$.
Similarly, we have $\Res{r}{1}\sqrt{Q_0(r)}=i\omega$, whose real part is positive, and hence $\psi_-$ is dominant on the Stokes curves flowing into $r=1$.
Then, once we choose how to draw the branch cuts, all the other dominance relations are automatically fixed since dominance relations, in general, alternate between adjacent Stokes curves.%
\footnote{The Stokes curve ``adjacent'' to the one across a branch cut lies on the next Riemann sheet.}
As shown in the left panel of figure \ref{fig:SL_schematic_Schwarzschild}, we draw the cuts originating from the branch of $\sqrt{Q_0}$ between $r_1$ and $r_4$, and between $r_2$ and $r_3$ (out of the figure),%
\footnote{The naming of the turning points, $r_{1,2,3,4}$, is arbitrary. We call the one in the ${\rm Re}\,(r) > 0 , \, {\rm Im} \,(r)>0$ region by $r_1$ and go around the origin counterclockwise to name the others in the increasing order, in the limit $\vert {\rm Im}\,(\omega) \vert \gg 1$.}
while those from the branch of $Q_0^{-1/4}$ (or $1/\sqrt{S_{\rm odd}}$) in eq.~\eqref{eq:WKB_solutions} between $r=0$ and $r=\infty$, and between $r=1$ and $r=\infty$, additional to the previous two.
With these choices of branches and cuts, all the dominance relations are determined as in the left panel of figure \ref{fig:SL_schematic_Schwarzschild}, denoted by \circled{$+$} and \circled{$-$}.

In order to impose the appropriate boundary conditions at the horizon ($r=1$) and the spatial infinity ($r=\infty$), we connect the exact WKB solutions constructed in the corresponding regions by analytic continuation. 
The resultant effect can be summarized by the pairs of the exact WKB solutions $\psi_\pm^{(1)}$ and $\psi_\pm^{(\infty)}$ in the two regions related by a matrix $M$, 
\begin{align}
    \begin{pmatrix}
        \psi_+\\
        \psi_-
    \end{pmatrix}^{(1)}
    =
    M
    \begin{pmatrix}
        \psi_+\\
        \psi_-
    \end{pmatrix}^{(\infty)}\, .
\end{align}
Due to the presence of the logarithmic spirals, the path that connects the two asymptotic points crosses the Stokes curves infinitely many times. Taking this phenomenon into account is essential to obtain the correct values of QNMs, and the Borel summability ensures a finite result even with the infinite crossings. 
These crossings can be decomposed into an infinite repetition of five effects. As observed in the left panel of figure \ref{fig:SL_schematic_Schwarzschild}, each set consists of a crossing of the Stokes curve emanating from one of the turning points (say $r_1$) into the regular singular point $r=1$, changing the integral bounds from $r_1$ to another turning point ($r_4$) as in eq.~\eqref{eq:connectionT}, a crossing of the Stokes curve emanating from $r_4$ into $r=1$, changing the integral bounds from $r_4$ back to $r_1$, and a crossing of the branch cut connecting $r=1$ and $r=\infty$, given by eq.~\eqref{eq:crossing_branchcut}.
After repeating these processes infinitely many times, the path crosses two Stokes curves that are attached to two different turning points ($r_1$ and $r_2$ in figure \ref{fig:SL_schematic_Schwarzschild}). Combining all these effects, we obtain the contour matrix $M$ as
\begin{align}
          M & = \lim_{n\to\infty} \left[ \left( V_- \right)^{-1} \Gamma_{14} \left( V_- \right)^{-1} \Gamma_{41} \, D_1 \right]^n V_+ \Gamma_{12} V_+ \Gamma_{21} \\
          & =
          \lim_{n\to\infty} 
          \begin{pmatrix}
            \displaystyle
            \left( \nu_1^+ \right)^n 
            \qquad
            & \displaystyle
            0 \vspace{2mm} \\
            \displaystyle
            -i \nu_1^+ \, \frac{\left( \nu_1^- \right)^n - \left( \nu_1^+ \right)^n}{\nu_1^- - \nu_1^+}(1 + u_{41}^2) \qquad
            & \left( \nu_1^- \right)^n
          \end{pmatrix}
          \begin{pmatrix}
            1 \quad
            & 
            \displaystyle
            i \left( 1 + u_{12}^2 \right) \vspace{2mm} \\
            0 \quad & 1
          \end{pmatrix}
          \\
        &=\lim_{n\to\infty}
          \begin{pmatrix}
            \displaystyle
            \left( \nu^+_1 \right)^n
            & \displaystyle
              i \left( \nu^+_1 \right)^n \left( 1 + u_{12}^2 \right) \vspace{2mm} \\
            \displaystyle
            -i \nu_1^+ \, \frac{\left( \nu_1^- \right)^n - \left( \nu_1^+ \right)^n}{\nu_1^- - \nu_1^+}(1 + u_{41}^2) \quad
            & \displaystyle
              \left( \nu^-_1 \right)^n + \nu_1^+ \, \frac{\left( \nu_1^- \right)^n - \left( \nu_1^+ \right)^n}{\nu_1^- - \nu_1^+} \left( 1 + u_{41}^2 \right) \left( 1 + u_{12}^2 \right)
          \end{pmatrix}\,.
          \label{eq:contourmatrix_Schwarzschild}
\end{align}
Since the set of effects $(V_-)^{-1} \Gamma_{14} (V_-)^{-1} \Gamma_{41} D_1$ repeats infinite times before the subsequent effects $V_+ \Gamma_{12} V_+ \Gamma_{21}$, the starting point of the contour should not matter. That is, matrix multiplication of any subset of $(V_-)^{-1} \Gamma_{14} (V_-)^{-1} \Gamma_{41} D_1$ from the left to \eqref{eq:contourmatrix_Schwarzschild} should be equally valid. Indeed, one can easily show that such a multiplication would change the exact form of $M$ but leaves the condition for QNMs the same.

To satisfy the boundary conditions, namely the incoming solution $\psi_-$ at $r=1$ and the outgoing one $\psi_+$ at $r=\infty$, we must demand $(M)_{22}=0$, i.e.,
\begin{align}
\lim_{n\to\infty}
    (\nu^-_1)^n \left[ 1+ \nu_1^+ \, \frac{1 - \left( \frac{\nu_1^+}{\nu_1^-} \right)^n}{\nu_1^- - \nu_1^+} \left( 1+u_{41}^2 \right) \left( 1+u_{12}^2 \right) \right]=0\,.
    \label{eq:M22_Schwarzschild}
\end{align}
Note that the overall factor $\left( \nu_1^- \right)^n$ does not affect the QNM condition \eqref{eq:M22_Schwarzschild} because $\left( \nu_1^- \right)^n\neq0$. When one is interested in quantities related to the normalization of mode functions, a more detailed treatment of the divergence is necessary, which is beyond the scope of this paper.
Since $\Res{r}{1}\,S_{\text{odd}}=i \eta \omega$ (see appendix~\ref{app:regular_residue}), we obtain
\begin{align}
\label{eq:coeff_cutcrossing}
    \nu^{\pm}_1=\exp\left[ i\pi \left( 1\pm2\Res{r}{1}S_{\text{odd}} \right) \right] = \exp\left( i\pi \mp 2\pi \eta \omega \right)\,.
\end{align}
Under our implicit condition ${\rm Re} \left( \omega \right) > 0$, we have $\mathrm{Re} \left( \nu_1^+/\nu_1^- \right) < 1$ and consequently $\big( \nu_1^+ / \nu_1^- \big)^n \to 0$ in the limit $n\to\infty$.
Then the condition \eqref{eq:M22_Schwarzschild} reduces to
\begin{align}
\label{eq:QNMcondition_Schwarzschild}
    \ee^{4\pi\omega} + u_{12}^2+u_{41}^2+u_{12}^2u_{41}^2=0\,.
\end{align}
We emphasize that the result is exact up to this point and does not rely on any approximation. Unless the Stokes geometry qualitatively changes in relation to the contour we consider, this is the full result.%
\footnote{This does not mean that eq.~\eqref{eq:QNMcondition_Schwarzschild} applies to any value of $\omega$, since a change of the parameter beyond a certain level may lead to a qualitative change of the Stokes geometry. The classification of the Stokes geometry in the entire parameter space is beyond the scope of our present study.}

We now make an approximation that is valid for higher overtones.
In the case of $\mathrm{Im} \left( \omega \right) < 0$ and $\vert \text{Im} ( \omega ) \vert \gg 1$, we find $u_{12}^2\simeq \ee^{i\pi 
\eta s},u_{41}^2\simeq \ee^{-i\pi \eta s}$ (see appendix \ref{app:phase_int}) and obtain
\begin{align}
    \omega \simeq \frac{\log {(1+2\cos{\pi s})}}{4\pi}-\frac{i}{2}\bigg(n+\frac{1}{2}\bigg) \; , 
    \label{eq:SchBH_QNM}
\end{align}
where $n$ a large positive integer and we have set $\eta = 1$.
Note that this formula implies that, for spin-$1$ perturbations $s=1$, $\omega$ is purely imaginary, and the higher overtones do not oscillate. One may argue that, if $\mathrm{Re} \, (\omega) = 0$, some of the premises in our calculation, e.g.~the hierarchy between $\nu_1^\pm$, would not hold, which is true, and lose the validity of the result. In fact, this is a special case in terms of the Stokes geometry where the logarithmic spirals appear to disappear. One way to handle this issue is to interpret the $s=1$ case as a limit of the (continuous) parameter in the vicinity of $s=1$. Provided that the Stokes geometry is non-singular around $s=1$, this gives the correct result, that is eq.~\eqref{eq:SchBH_QNM}.

For higher overtones with $\text{Re}\,(\omega)<0$, we can calculate the asymptotic QNMs in the same manner. Under the choice of branch as in \eqref{eq:branch_Schwarzchild} and branch cuts similar to the left panel of figure \ref{fig:SL_schematic_Schwarzschild}, the only changes we need to take into account as compared with the case ${\rm Re} \, (\omega)$ are: (i) the Stokes geometry of $\text{Re}~(\omega) < 0$ is the upside-down of that of $\text{Re}~(\omega)>0$ in the limit $\vert \textrm{Im} \, (\omega) \vert \to \infty$, (ii) the directions to cross Stokes curve along the contour are all flipped, and (iii) we now have $\vert \nu_1^- / \nu_1^+ \vert < 1$ and thus $\left( \nu_1^- / \nu_1^+ \right)^n \to 0$, see the expression \eqref{eq:coeff_cutcrossing}. As a result, the QNM condition is slightly different from \eqref{eq:SchBH_QNM}, and we obtain
\begin{align}
\label{eq:QNMcondition_Sch_negative}
    \ee^{-4\pi\omega} + u_{12}^2 + u_{41}^2 + u_{12}^2 u_{41}^2 = 0 \; ,
\end{align}
where the turning points are numbered in the same manner as in the left panel of figure \ref{fig:SL_schematic_Schwarzschild}.
This condition indeed gives the value of $\omega$ with the correct real part, that is, in the limit $\mathrm{Im} \, (\omega) \to - \infty$,
\begin{align}
\label{eq:SchBH_QNM_negativefreq}
    \omega \simeq - \frac{\log \left( 1 + 2 \cos \pi s \right)}{4\pi} - \frac{i}{2} \left( n + \frac{1}{2} \right) \; ,
\end{align}
where $n$ is again a large positive integer and $\eta =1$ is taken.
eqs.~\eqref{eq:SchBH_QNM} and \eqref{eq:SchBH_QNM_negativefreq} are the main result of this paper.

Substituting $s=2$, eqs.~(\ref{eq:SchBH_QNM}) and \eqref{eq:SchBH_QNM_negativefreq} match with the numerical calculations \cite{Nollert:1993zz,Hod:1998vk} and analytic calculations \cite{Motl:2002hd,Motl:2003cd,Andersson:2003fh}.
Let us emphasize that, if we did not properly include the effects of the spiral flows of the Stokes curve at the horizon, the real part of asymptotic QNMs would become a wrong value (e.g., by taking $n=1$ in \eqref{eq:M22_Schwarzschild}, the real part becomes $\log 4 /(4\pi)$ with $s=2$).

Let us comment on the case of spin-$0$ perturbations. 
As mentioned below eq.~\eqref{eq:V_Schwarzschild}, the Stokes geometry is qualitatively different when $s=0$. In particular, as can be observed from the explicit expression of $Q_0$ in eq.~\eqref{eq:Q0}, $r=0$ is no longer a regular singular point but a simple pole, leaving $r=1$ as the only regular singular point, while the number of the turning points reduces to $3$. This change can be understood as the limit of first taking $\mathrm{Im} \, (\omega) \to -\infty$ and then sending $s \to 0$. Under this interpretation, the two of the turning points, called $r_1$ and $r_4$ in the left panel of figure \ref{fig:SL_schematic_Schwarzschild}, merge and form a single turning point on the real axis. Then keeping in mind that the Stokes curve emanating from this turning point and spiraling into the regular singular point $r=1$ is two-fold, essentially the same calculation as above follows, resulting in the correct result of the form \eqref{eq:SchBH_QNM} and \eqref{eq:SchBH_QNM_negativefreq} with $s=0$, which can be verified in the literature \cite{Motl:2002hd}.%
\footnote{As a matter of fact, this derivation for the $s=0$ case is not mathematically satisfactory. This case may as well be obtained as the limit of taking $s \to 0$ with $\omega$ fixed. However, in this case, one of the turning points and a regular singular point, instead of two turning points, merge, making $r=0$ a simple pole. Starting from this Stokes geometry, the condition for QNMs appears to be different. This means that the order of taking limits would matter, and a clear criterion for the correct order seems lacking a priori. }

In our derivation, we notice that the singular points of the differential equation play an important role in determining the higher overtone frequencies. More specifically, the residue of $S_{\rm odd}$ at $r=1$ is how $4\pi\omega$ enters in eqs.~\eqref{eq:QNMcondition_Schwarzschild} and \eqref{eq:QNMcondition_Sch_negative}. The origin of this term is $\nu_1^- / \nu_1^+$, whereas $\nu_1^\pm$ represents the effects of the crossing of branch cuts that, in turn, pick up the residue at the event horizon.
The other terms containing $u_{12}$ and $u_{41}$ in general depend on $\omega$ as well; however, as $\vert {\rm Im} \, (\omega) \vert$ becomes larger, all the four turning points get closer to the origin. Moreover, assuming $\vert {\rm Re} \, (\omega) \vert \ll \vert {\rm Im} \, (\omega) \vert$, they are equally spaced in angle by $90^\circ$ around $r=0$, which can be understood by the fact that the turning points can be found approximately by solving $\omega^2 r^4 - s^2 \simeq 0$, giving $r \simeq {\rm e}^{i \pi \left( 2n + 1 \right)/4} \left( - s^2 / \omega^2 \right)^{1/4}$ with $n=0,1,2,3$ and approximately $\omega^2 < 0$ \cite{Andersson:2003fh}. In this case, the two closed contour integrals $u_{12}$ and $u_{14}$, which encircle branch cuts between turning points, give roughly equal contributions, each of which is half of the sum of the residues on the complex $r$ plane. The contributions from the residue at $r=0$ and that at $r=\infty$ cancel out each other, making $u_{12}$ and $u_{14}$ given only by the residue at $r=0$ and independent of $\omega$. This nontrivial dependence on the singularity not only at the event horizon but also at the curvature singularity (and the spatial infinity) is due to the analytic structure of the Schwarzschild background spacetime.
Furthermore, the factor $\log \left( 1 + 2 \cos \pi s \right)$ is originated from the nature of regular singularity at the event horizon, and taking into account the infinite number of logarithmic spirals is a crucial ingredient to obtain the correct QNMs.

\section{Conclusion and discussion}
\label{sec:conclusion}
We have studied the black hole QNMs using the exact WKB analysis.
To our knowledge, this is the first concrete application of the exact WKB analysis to the black hole perturbations.
The method utilizes the Borel resummation of the standard WKB series to examine the global behavior of solutions to the perturbation equations. The fundamental ingredients in our analysis are: (i) the Stokes phenomena, i.e.~discontinuos changes of asymptotic expansions (with the full solutions stay continuous), (ii) the phase integrals between turning points, and (iii) the residue of the WKB series at the event horizon, which arises when crossing branch cuts on the complex coordinate plane. The effects of (i) are summarized in eq.~\eqref{eq:connection_matrix} and called the Voros connection formulae. The phase integrals (ii) are given by eq.~\eqref{eq:connectionT}, which amounts to the evaluation of closed contour integrals around branch cuts as can be seen in eq.~\eqref{eq:uij} and figure \ref{fig:gamma_ij}. The singular points of the differential equation also play an important role in (iii), where a branch cut inherits the consequence of that singular point which it emerges from. In this context, the exact WKB technique reduces a problem of differential equations to that of complex functional analysis, and those different effects are related to each other thanks to the analyticity of the complex function $Q(x,\eta)$. The discontinuous and/or singular behaviors capture the non-perturbative nature of the black hole background geometry.

In the exact WKB analysis, the analytic continuation of the WKB solutions is performed across the Stokes curves.
In this work, we have developed a methodology to deal with the logarithmic spiral flow of a Stokes curve near a regular singularity, which is essential to properly obtain QNM frequencies. By incorporating the spiral into the analytic continuation, we derive the QNM conditions for the Morse potential and the corresponding QNM frequencies correctly, before applying our method to the QNMs of the Schwarzschild black hole. The correct asymptotic values of higher overtone QNM frequencies are then successfully reproduced.
We note that, in our derivation, the solutions are analytically continued from the horizon to the spatial infinity, where the incoming and outgoing, respectively, behaviors are assumed as the commonly accepted boundary conditions for black hole perturbations.

The main difference between our approach and another widely known method, called the monodromy \cite{Motl:2003cd,Andersson:2003fh}, is that
our path of the analytic continuation is taken from the horizon to infinity along the positive real axis of radial coordinate, while the monodromy, as its name stands, follows a closed contour on the complex plane instead \cite{Andersson:1995zk,Motl:2003cd,Andersson:2003fh}.  
Our contour is the one taken in numerical computations and can be identified in a straightforward manner with no ambiguity because the QNM eigenfunctions we look for are originally defined on the real axis of the radial coordinate.
Also, our method enables the analytic continuation without any approximation at least formally once the Stokes geometry is known, which is of benefit for a wider range of applications over a typical procedure of the monodromy technique. 

In our exact WKB methodology, the QNM conditions themselves involve no approximations and can be applied to lower overtones.
Notably, the Stokes geometries exhibit significant changes at lower overtones.
The physical implications of such nontrivial transformations are intriguing and warrant further investigation.
Of course, calculating the phase integrals poses challenges, as it requires handling the Borel summation of $S_{\text{odd}}$ and performing contour integrals on a genus-1 Riemann surface.
We have constructed a method in appendix \ref{app:phase_int} to compute the leading-order results of the phase integrals that is valid for higher overtones, but further developing methods for the analytical or numerical evaluation of more general phase integrals is an important direction for future research.

There remain some subtleties in taking some limits. In the case of the Schwarzschild QNMs, the spin-$0$ perturbation corresponds to a special point in the parameter space in terms of the Stokes geometry, for which one of the turning points merges and the curvature singularity $r=0$ reduces from a pole of order $2$ to that of order $1$. However, the merging can be interpreted as being either with another turning point or with the double pole at $r=0$. This corresponds to an ambiguity in the order of taking limits, where there appear no criteria to determine which one is the ``correct'' order. A seemingly similar issue is known to occur when considering the higher overtones of perturbations around a charged black hole and taking the limit of vanishing charge \cite{Motl:2002hd,Motl:2003cd}. Taking special limits is equivalent to handling topologically different Stokes geometries and may be related to recovering certain symmetries in some cases. A further investigation of such limits would be an interesting venue to understand the physical system under consideration.

The exact WKB analysis is not limited to general relativity; it is also applicable to gravitational theories beyond general relativity. One of its key advantages is that it does not require the knowledge of special functions to perform analytic continuation. Recently, QNMs in gravitational theories beyond general relativity have been actively studied, e.g., in \cite{Cardoso:2019mqo,McManus:2019ulj,Hirano:2024fgp,Moreira:2023cxy,Moura:2021nuh,Moura:2022gqm}.
We believe that our formulation provides a powerful tool for exploring the asymptotic behavior of QNMs in such gravity theories.

Our formulation also offers a clear framework for calculating reflection and transmission coefficients.
This suggests that the exact WKB analysis may also be applicable to the study of black hole graybody factors \cite{Hawking:1975vcx}\footnote{For the application of the graybody factors to the model of ringdown and QNM excitation, see Refs. \cite{Oshita:2022pkc,Oshita:2023cjz,Okabayashi:2024qbz,Konoplya:2024lir}.} or QNM excitation factors \cite{Leaver:1986gd,Sun:1988tz,Andersson:1995zk,Glampedakis:2001js,Glampedakis:2003dn,Berti:2006wq,Zhang:2013ksa,Oshita:2021iyn} as they are defined by the asymptotic amplitude of the homogeneous solutions.
In summary, the exact WKB analysis is a versatile and powerful tool for addressing black hole perturbation theories across a wide range of contexts.

\acknowledgments
T.M.~and R.N.~are particularly grateful to Takashi Aoki for his intensive lectures on the exact WKB, which gave us confidence for the direction of our project, for sharing his Mathematica code to draw Stokes curves, and for helpful comments on the details of the exact WKB analysis.
R.N.~thanks Masashi Kimura for his useful comments during their discussion.
T.M.~was supported by JSPS KAKENHI Grant Number JP23KJ1543.
R.N.~was in part supported by MEXT KAKENHI Grant Number JP23K25868.
H. O.~was supported by JSPS KAKENHI Grant Numbers JP23H00110 and Yamada Science Foundation. 
N.O.~was supported by Japan Society for the Promotion of Science (JSPS) KAKENHI Grant No.~JP23K13111 and by the Hakubi project at Kyoto University.

\appendix

\section{The residue of $S_{\text{odd}}$ at a double pole (regular singular point)}
\label{app:regular_residue}
There is a convenient theorem about the residue of $S_{\text{odd}}$ evaluated at a regular singular point of the corresponding differential equation of the form \eqref{eq:Schrodinger_type_eq} (i.e.~a double pole of $Q(x,\eta)$), which states as follows \cite{kawai2005algebraic}. 
\begin{prop} \label{prop:regular_residue}
    Suppose $Q(x,\eta)$ has a double pole at $x=b$ and is expanded by large $\eta$ as
    \begin{align}
        Q(x,\eta)=Q_0(x) + \eta^{-1}Q_1(x) + \cdots + \eta^{-N}Q_N(x) \; ,
    \end{align}
    so that each $Q_j$ does not have a pole higher than double at $x=b$.
    Under these assumptions, $S_{\text{odd}}$ has a simple pole at $x = b$, and its residue at $x=b$ is given as
    \begin{align}
        \underset{x=b}{\text{Res}}\,S_\text{odd} = c\eta\sqrt{1+\frac{1}{4c^2\eta^2}} \; ,
        \label{eq:ResidueOfSodd}
    \end{align}
    where $c = \underset{x=b}{\text{Res}}\sqrt{Q} \,$.
\end{prop}

\begin{proof}
    Under the given assumptions, the recursion relation \eqref{eq:S_j} infers that each $S_j^{\pm}$ defined in eq.~\eqref{eq:def_Sjpm} has a simple pole at $x=b$ and no higher-oder poles. Then each $S_j^{\pm}$ can be written as
    \begin{align}
        S_j^{\pm} = \sum_{k\geq -1}f_{j,k}^{(\pm)}(x-b)^k \; ,
    \end{align}
    and then we have
    \begin{align}
        S^{\pm} = \sum_{j = -1}^\infty \eta^{-j} S^\pm_j 
        = \sum_{j,k\geq -1} \eta^{-j} f_{j,k}^{(\pm)}(x-b)^k \; .
    \end{align}
    Substituting this equation into \eqref{eq:Riccati} and looking at the coefficient of the most singular term $(x-b)^{-2}$, we obtain
    \begin{align}
        \bigg(\sum_{j\geq -1}\eta^{-j} f_{j,-1}^{(\pm)}\bigg)^2 - \bigg(\sum_{j\geq -1}\eta^{-j} f_{j,-1}^{(\pm)}\bigg) -\eta^2c^2=0 \; ,
    \end{align}
    whose solution then gives
    \begin{align}
         F^{(\pm)}:=\sum_{j\geq -1}\eta^{-j} f_{j,-1}^{(\pm)}=\frac{1}{2}\pm\frac{1}{2}\sqrt{1+4\eta^2c^2} \; .
    \end{align}
    Notice that $F^{(\pm)}$ are the residues of $S^\pm$ at $x=b$. Recalling from the definition in eq.~\eqref{eq:Spmoddeven} that $S_{\rm odd} = \left( S^+ - S^- \right)/2$, we finally find
    \begin{align}
        \underset{\rho=b}{\text{Res}}\,S_\text{odd}
        =\frac{1}{2}(F^{(+)} - F^{(-)}) = \frac{1}{2}\sqrt{1+4\eta^2c^2} \; ,
    \end{align}
    arriving at the desired equality.
\end{proof}

In practice, once $Q(x,\eta)$ is given, one can compute each $S_{2j+1}$ explicitly and look for the coefficients of the $(x-b)^{-1}$ terms, and $\mathop{\rm Res} S_{\rm odd}$ is their sum. When $Q$ contains a double (or higher) pole, the higher order in $S_j$, the less singular at $x=b$, as can be seen from the recursion relation \eqref{eq:recursion}. Thus typically it suffices to compute the first few terms of $S_{2j+1}$ in order to find the residue.

\section{Calculation of the phase integral}
\label{app:phase_int}

In this appendix, we calculate the integrals of $u_{12}$ and $u_{41}$ in section~\ref{sec:Schwarzschild}. In this appendix, we rescale the radial coordinate $r$ by the Schwarzschild radius $r_S = 2M$ for convenience.
Firstly, we approximate the integrand as
\begin{align}
    I_{ij}
    := \frac{1}{2} \oint_{\gamma_{ij}} S_{\text{odd}}~\dd r
    \simeq \eta \int_{r_i}^{r_j}\sqrt{Q_0(r)} ~\dd r \; .
\end{align}
While the original integral contour is a closed path $\gamma_{ij}$ as shown in figure \ref{fig:gamma_ij}, $Q_0$ is not singular at any turning points, and thus the integral on the right-hand side can be evaluated by a definite integration from the turning point $\gamma_i$ to the other $\gamma_j$.
Provided that the asymptotic behavior of $S_{\rm odd}$ is well approximated by the leading-order WKB term $Q_0$, the above approximation is valid when $r_i$ and $r_j$ are both in the close vicinity of a singular point. In the application to the Schwarzschild QNMs in section \ref{sec:Schwarzschild}, we observe that, the larger $\vert \omega \vert$ is, the closer all the turning points to the origin $r=0$. Hence the above approximation is justified for computations of higher overtones.

We formally expand $\sqrt{Q_0}$ for large $\vert \omega \vert$ as
\begin{align}
    \sqrt{Q_0}&=i\omega \, \frac{r}{r-1}\sqrt{1-J}
    =i\omega \, \frac{r}{r-1}\bigg(1-\frac{1}{2}J-\frac{1}{8}J^2 + \cdots\bigg) 
    =i\omega \, \frac{r}{r-1}\sum_{n=0}^{\infty}a_nJ^n 
    \label{eq:expandQ0}
\end{align}
where
\begin{align}
    J(r)&=\frac{1}{\omega^2}\frac{r-1}{r^4} \left[ l \left( l+1 \right) r-s^2 \right] \,,\\
    a_n&=-\frac{\Gamma(n-1/2)}{2\sqrt{\pi}n!}\,.
\end{align}
Also, in the case of higher overtones of the Schwarzschild QNMs, the turning points can be expanded for large $|\omega|$ as
\begin{align}
    r_j=\theta^{j-1}\sqrt{\frac{s}{\omega}}+\mathcal{O} \big( \vert \omega \vert^{-1} \big) \; ,
    \qquad \theta=\ee^{i\pi/2} \; ,
\end{align}
where the subscript $j=1,2,3,4$ denotes each turning point and 
we take $\arg \omega \simeq -\pi/2$ for $\vert {\rm Im} \, (\omega) \vert \gg \vert {\rm Re} \, (\omega) \vert$.
Note that $r_j = \mathcal{O}\big( \vert \omega \vert^{-1/2} \big)$, and $J(r)$ can be approximated as $J(r) \simeq s^2 / \left( \omega^2 r^4 \right) = \mathcal{O}\big( \vert \omega \vert^0 \big)$ in the vicinity of the turning points.

Let us consider integrating \eqref{eq:expandQ0} term by term, substituting $r$ by the turning points $r_j$ into the integrated formula, and expanding it in terms of $1/\omega$.
The integration of the first term in eq.~\eqref{eq:expandQ0}, $A_0$, then reads
\begin{align}
    A_0 & := \int^{r_j} i\omega \, \frac{r}{r-1}\,\dd r \notag\\
    &=i\omega \left[ r + \log \left( r-1 \right) \right] \Big\vert_{r = r_j } \notag\\
    &=i\omega \left[ i\pi -\frac{1}{2} \, r^2-\frac{1}{3} \, r^3+O\big( r^4 \big) \right]\bigg\vert_{r = r_j} \notag\\
    &= -\pi\omega 
    -i\theta^{2(j-1)} \, \frac{s}{2}
    +\mathcal{O}\big( \vert \omega \vert^{-1/2} \big) \; .
\end{align}
The leading order is thus $\mathcal{O}\big( \vert \omega \vert^0 \big)$.
For $n\geq 1$, it is sufficient to integrate the following function to obtain the $\mathcal{O}\big( \vert \omega \vert^0 \big)$ contributions:
\begin{align}
    i \omega \, \frac{r}{r-1} \, a_n J^n
    \simeq - i a_n \, \omega^{1-2n} s^{2n} r^{1-4n} \; .
\end{align}
The integration of this quantity is given by
\begin{align}
    A_n := \int^{r_j} i \omega \, \frac{r}{r-1} \, a_n J^n \, \dd r
    \simeq 
    i s \, \theta^{2 \left( j-1 \right)} \, \frac{a_n}{4n-2}
    + \mathcal{O} \big( \vert \omega \vert^{-1/2} \big)\,.
\end{align}
Since the infinite series of $a_n/(4n-2)$ over $n=0,1,2,\dots$ is convergent and gives $-\pi/4$, 
we obtain, summing the leading-order terms in $A_n$ over $n$ $(n=0,1,2,\ldots)$, 
\begin{align}
    I_j
    &:=\int^{r_j}\sqrt{Q_0}\, \dd r \notag\\
    &=\sum_{n=0}^{\infty} A_n +\mathcal{O} \big( \vert \omega \vert^{-1/2} \big) \notag\\
    &=-\pi\omega -\frac{i\pi s}{4} \, \theta^{2(j-1)}
    +\mathcal{O} \big( \vert \omega \vert^{-1/2} \big)\,.
\end{align}
Finally, using this result, $I_{12}$ and $I_{41}$ are evaluated as
\begin{align}
    I_{12}&=I_2-I_1=
    \frac{i\pi s}{2}
    +\mathcal{O}\big( \vert \omega \vert^{-1/2} \big) \,, \\
    I_{41}&=I_1-I_4=
    -\frac{i\pi s}{2}
    +\mathcal{O}\big( \vert \omega \vert^{-1/2} \big) \,.
\end{align}
These approximate values are used to obtain the final expressions of the Schwarzschild QNMs in the main text.

\bibliographystyle{JHEP}
\bibliography{manuscript}
 
\end{document}